\documentclass[a4paper,onecolumn, unpublished]{compositionalityarticle}
\pdfoutput=1

\usepackage[utf8]{inputenc}
\usepackage{amsmath}
\usepackage{amssymb}
\usepackage{bbold}
\usepackage{comment}
\usepackage{graphicx}
\usepackage{scalerel}
\usepackage{braket}
\usepackage{leftidx}
\usepackage[title]{appendix}
\usepackage{tikz}
\usepackage{tikzit}
\usepackage{tikz-cd}
\usepackage{geometry}
\usepackage{mathtools}
\usepackage{ stmaryrd }
\usepackage{circuitikz}
\usepackage{textcomp}
\usepackage{accents}
\usepackage{comment}
\usepackage{ mathrsfs }
\usepackage{authblk}
\usepackage{breakurl,amsmath,amsthm,amssymb,xspace,xypic,enumerate,color}
\usepackage{float}

\tikzcdset{scale cd/.style={every label/.append style={scale=#1},
    cells={nodes={scale=#1}}}}

\usetikzlibrary{shapes.multipart}

\tikzstyle{bwSpider}=[
       rectangle split,
       rectangle split parts=2,
       rectangle split part fill={black,white},
 minimum size=3.6 mm, inner sep=-2mm, draw=black,scale=0.5,rounded corners=0.8 mm
       ]
 \tikzstyle{wbSpider}=[
       rectangle split,
       rectangle split parts=2,
       rectangle split part fill={white,black},
 minimum size=3.6 mm, inner sep=-2mm, draw=black,scale=0.5,rounded corners=0.8 mm
       ]
\tikzstyle{cWire}=[densely dotted, thick]
\tikzstyle{env}=[copoint,regular polygon rotate=0,minimum width=0.2cm, fill=black]

\tikzstyle{probs}=[shape=semicircle,fill=white,draw=black,shape border rotate=180,minimum width=1.2cm]

\tikzstyle{every picture}=[baseline=-0.25em,scale=0.5]
\tikzstyle{dotpic}=[] % for backwards-compatibility
\tikzstyle{diredges}=[every to/.style={diredge}]
\tikzstyle{math matrix}=[matrix of math nodes,left delimiter=(,right delimiter=),inner sep=2pt,column sep=1em,row sep=0.5em,nodes={inner sep=0pt},text height=1.5ex, text depth=0.25ex]

\tikzstyle{inline text}=[text height=1.5ex, text depth=0.25ex,yshift=0.5mm]
\tikzstyle{label}=[font=\footnotesize,text height=1.5ex, text depth=0.25ex,yshift=0.5mm]
\tikzstyle{left label}=[label,anchor=east,xshift=1.5mm]
\tikzstyle{right label}=[label,anchor=west,xshift=-1.5mm]

\tikzstyle{braceedge}=[decorate,decoration={brace,amplitude=2mm,raise=-1mm}]
\tikzstyle{small braceedge}=[decorate,decoration={brace,amplitude=1mm,raise=-1mm}]

\tikzstyle{doubled}=[line width=1.6pt] % set the line width for all doubled (quantum) maps/wires
\tikzstyle{boldedge}=[doubled,shorten <=-0.17mm,shorten >=-0.17mm]
\tikzstyle{boldedgegray}=[doubled,gray,shorten <=-0.17mm,shorten >=-0.17mm]
\tikzstyle{singleedgegray}=[gray]%,shorten <=-0.1mm,shorten >=-0.1mm]

\tikzstyle{semidoubled}=[line width=1.4pt] % set the line width for all doubled (quantum) maps/wires
\tikzstyle{semiboldedgegray}=[semidoubled,gray,shorten <=-0.17mm,shorten >=-0.17mm]

\tikzstyle{boxedge}=[semiboldedgegray]

\tikzstyle{dottededge}=[dashed,shorten <=-0.17mm,shorten >=-0.17mm]

\tikzstyle{bolddotteddashed}=[very thick,dashed,shorten <=-0.17mm,shorten >=-0.17mm]
\tikzstyle{vboldedgedashed}=[doubled,dashed,shorten <=-0.17mm,shorten >=-0.17mm]
\tikzstyle{left hook arrow}=[left hook-latex]
\tikzstyle{right hook arrow}=[right hook-latex]
\tikzstyle{sembracket}=[line width=0.5pt,shorten <=-0.07mm,shorten >=-0.07mm]

\tikzstyle{causal edge}=[->,thick,gray]
\tikzstyle{causal nondir}=[thick,gray]
\tikzstyle{timeline}=[thick,gray, dashed]

\tikzstyle{cedge}=[<->,thick,gray!70!white]

\tikzstyle{empty diagram}=[draw=gray!40!white,dashed,shape=rectangle,minimum width=1cm,minimum height=1cm]
\tikzstyle{empty diagram small}=[draw=gray!50!white,dashed,shape=rectangle,minimum width=0.6cm,minimum height=0.5cm]

\tikzstyle{dot}=[inner sep=0mm,minimum width=2mm,minimum height=2mm,draw,shape=circle]

\tikzstyle{phase dot}=[pdot,phase dimensions]
\tikzstyle{wphase dot}=[dot, phase dimensions]

\tikzstyle{leak}=[white dot, shape=regular polygon, minimum size=300mm, regular polygon sides=3, outer sep=-0.2mm, regular polygon rotate=270]
\tikzstyle{preleak}=[trapezium, trapezium angle=67.5, draw, inner sep=0.1pt, outer sep=0pt, minimum height=2mm, minimum width=4pt,rotate=270]
\tikzstyle{proj}=[white dot, shape=regular polygon, minimum size=3.3 mm, regular polygon sides=4, outer sep=-0.2mm]
\tikzstyle{Vleak}=[white dot, shape=regular polygon, minimum size=3.3 mm, regular polygon sides=3, outer sep=-0.2mm, regular polygon rotate=90]
\tikzstyle{dleak}=[white dot, line width=1.6pt, shape=regular polygon, minimum size=3.3 mm, regular polygon sides=3, outer sep=-0.2mm, regular polygon rotate=270]

\tikzstyle{Wsquare}=[white dot, shape=regular polygon, rounded corners=0.8 mm, minimum size=3.3 mm, regular polygon sides=3, outer sep=-0.2mm]
\tikzstyle{Wsquareadj}=[white dot, shape=regular polygon, rounded corners=0.8 mm, minimum size=3.3 mm, regular polygon sides=3, outer sep=-0.2mm, regular polygon rotate=180]
\tikzstyle{ddot}=[inner sep=0mm, doubled, minimum width=2.5mm,minimum height=2.5mm,draw,shape=circle]

\tikzstyle{black dot}=[dot,fill=black]
\tikzstyle{white dot}=[dot,fill=white,,text depth=-0.2mm]
\tikzstyle{white Wsquare}=[Wsquare,fill=gray,,text depth=-0.2mm]
\tikzstyle{white Wsquareadj}=[Wsquareadj,fill=white,,text depth=-0.2mm]
\tikzstyle{green dot}=[white dot] % for backwards-compatibility
\tikzstyle{gray dot}=[dot,fill=gray!40!white,,text depth=-0.2mm]
\tikzstyle{red dot}=[gray dot] % for backwards-compatibility

\tikzstyle{black ddot}=[ddot,fill=black]
\tikzstyle{white ddot}=[ddot,fill=white]
\tikzstyle{gray ddot}=[ddot,fill=gray!40!white]

\tikzstyle{gray edge}=[gray!60!white]

\tikzstyle{small dot}=[inner sep=0.5mm,minimum width=0pt,minimum height=0pt,draw,shape=circle]

\tikzstyle{small black dot}=[small dot,fill=black]
\tikzstyle{small white dot}=[small dot,fill=white]
\tikzstyle{small gray dot}=[small dot,fill=gray!40!white]

\tikzstyle{causal dot}=[inner sep=0.4mm,minimum width=0pt,minimum height=0pt,draw=white,shape=circle,fill=gray!40!white]

\tikzstyle{phase dimensions}=[minimum size=5mm,font=\footnotesize,rectangle,rounded corners=2.5mm,inner sep=0.2mm,outer sep=-2mm]
\tikzstyle{dphase dimensions}=[minimum size=5mm,font=\footnotesize,rectangle,rounded corners=2.5mm,inner sep=0.2mm,outer sep=-2mm]
\tikzstyle{white phase dot}=[dot,fill=white,phase dimensions]
\tikzstyle{white phase ddot}=[ddot,fill=white,dphase dimensions]

\tikzstyle{white rect ddot}=[draw=black,fill=white,doubled,minimum size=5mm,font=\footnotesize,rectangle,rounded corners=2.5mm,inner sep=0.2mm]
\tikzstyle{gray rect ddot}=[draw=black,fill=gray!40!white,doubled,minimum size=6mm,font=\footnotesize,rectangle,rounded corners=3mm]

\tikzstyle{gray phase dot}=[dot,fill=gray!40!white,phase dimensions]
\tikzstyle{gray phase ddot}=[ddot,fill=gray!40!white,dphase dimensions]
\tikzstyle{grey phase dot}=[gray phase dot]
\tikzstyle{grey phase ddot}=[gray phase ddot]

\tikzstyle{small phase dimensions}=[minimum size=4mm,font=\tiny,rectangle,rounded corners=2mm,inner sep=0.2mm,outer sep=-2mm]
\tikzstyle{small dphase dimensions}=[minimum size=4mm,font=\tiny,rectangle,rounded corners=2mm,inner sep=0.2mm,outer sep=-2mm]

\tikzstyle{small gray phase dot}=[dot,fill=gray!40!white,small phase dimensions]
\tikzstyle{small gray phase ddot}=[ddot,fill=gray!40!white,small dphase dimensions]

\tikzstyle{small map}=[draw,shape=rectangle,minimum height=4mm,minimum width=4mm,fill=white]

\tikzstyle{cnot}=[fill=white,shape=circle,inner sep=-1.4pt]

\tikzstyle{asym hadamard}=[fill=white,draw,shape=NEbox,inner sep=0.6mm,font=\footnotesize,minimum height=4mm]
\tikzstyle{asym hadamard conj}=[fill=white,draw,shape=NWbox,inner sep=0.6mm,font=\footnotesize,minimum height=4mm]
\tikzstyle{asym hadamard dag}=[fill=white,draw,shape=SEbox,inner sep=0.6mm,font=\footnotesize,minimum height=4mm]

\tikzstyle{hadamard}=[fill=white,draw,inner sep=0.6mm,font=\footnotesize,minimum height=4mm,minimum width=4mm]
\tikzstyle{small hadamard}=[fill=white,draw,inner sep=0.6mm,minimum height=1.5mm,minimum width=1.5mm]
\tikzstyle{small hadamard rotate}=[small hadamard,rotate=45]
\tikzstyle{dhadamard}=[hadamard,doubled]
\tikzstyle{small dhadamard}=[small hadamard,doubled]
\tikzstyle{small dhadamard rotate}=[small hadamard rotate,doubled]
\tikzstyle{antipode}=[white dot,inner sep=0.3mm,font=\footnotesize]

\tikzstyle{scalar}=[diamond,draw,inner sep=0.5pt,font=\small]
\tikzstyle{dscalar}=[diamond,doubled, draw,inner sep=0.5pt,font=\small]

\tikzstyle{small box}=[rectangle,inline text,fill=white,draw,minimum height=5mm,yshift=-0.5mm,minimum width=5mm,font=\small]
\tikzstyle{small gray box}=[small box,fill=gray!30]
\tikzstyle{medium box}=[rectangle,inline text,fill=white,draw,minimum height=5mm,yshift=-0.5mm,minimum width=10mm,font=\small]
\tikzstyle{square box}=[small box] % for backwards-compatibility
\tikzstyle{medium gray box}=[small box,fill=gray!30]
\tikzstyle{semilarge box}=[rectangle,inline text,fill=white,draw,minimum height=5mm,yshift=-0.5mm,minimum width=12.5mm,font=\small]
\tikzstyle{large box}=[rectangle,inline text,fill=white,draw,minimum height=5mm,yshift=-0.5mm,minimum width=15mm,font=\small]
\tikzstyle{large gray box}=[small box,fill=gray!30]

\tikzstyle{Bayes box}=[rectangle,fill=black,draw, minimum height=3mm, minimum width=3mm]

\tikzstyle{gray square point}=[small box,fill=gray!50]

\tikzstyle{dphase box white}=[dhadamard]
\tikzstyle{dphase box gray}=[dhadamard,fill=gray!50!white]
\tikzstyle{phase box white}=[hadamard]
\tikzstyle{phase box gray}=[hadamard,fill=gray!50!white]

\tikzstyle{point}=[regular polygon,regular polygon sides=3,draw,scale=0.75,inner sep=-0.5pt,minimum width=9mm,fill=white,regular polygon rotate=180]
\tikzstyle{point nosep}=[regular polygon,regular polygon sides=3,draw,scale=0.75,inner sep=-2pt,minimum width=9mm,fill=white,regular polygon rotate=180]
\tikzstyle{copoint}=[regular polygon,regular polygon sides=3,draw,scale=0.75,inner sep=-0.5pt,minimum width=9mm,fill=white]
\tikzstyle{dpoint}=[point,doubled]
\tikzstyle{dcopoint}=[copoint,doubled]

\tikzstyle{pointgrow}=[shape=cornerpoint,kpoint common,scale=0.75,inner sep=3pt]
\tikzstyle{pointgrow dag}=[shape=cornercopoint,kpoint common,scale=0.75,inner sep=3pt]

\tikzstyle{wide copoint}=[fill=white,draw,shape=isosceles triangle,shape border rotate=90,isosceles triangle stretches=true,inner sep=0pt,minimum width=1.5cm,minimum height=6.12mm]
\tikzstyle{wide point}=[fill=white,draw,shape=isosceles triangle,shape border rotate=-90,isosceles triangle stretches=true,inner sep=0pt,minimum width=1.5cm,minimum height=6.12mm,yshift=-0.0mm]
\tikzstyle{wide point plus}=[fill=white,draw,shape=isosceles triangle,shape border rotate=-90,isosceles triangle stretches=true,inner sep=0pt,minimum width=1.74cm,minimum height=7mm,yshift=-0.0mm]

\tikzstyle{wide dpoint}=[fill=white,doubled,draw,shape=isosceles triangle,shape border rotate=-90,isosceles triangle stretches=true,inner sep=0pt,minimum width=1.5cm,minimum height=6.12mm,yshift=-0.0mm]

\tikzstyle{tinypoint}=[regular polygon,regular polygon sides=3,draw,scale=0.55,inner sep=-0.15pt,minimum width=6mm,fill=white,regular polygon rotate=180]

\tikzstyle{white point}=[point]
\tikzstyle{white dpoint}=[dpoint]
\tikzstyle{green point}=[white point] % for backwards-compatibility
\tikzstyle{white copoint}=[copoint]
\tikzstyle{gray point}=[point,fill=gray!40!white]
\tikzstyle{gray dpoint}=[gray point,doubled]
\tikzstyle{red point}=[gray point] % for backwards-compatibility
\tikzstyle{gray copoint}=[copoint,fill=gray!40!white]
\tikzstyle{gray dcopoint}=[gray copoint,doubled]

\tikzstyle{white point guide}=[regular polygon,regular polygon sides=3,font=\scriptsize,draw,scale=0.65,inner sep=-0.5pt,minimum width=9mm,fill=white,regular polygon rotate=180]

\tikzstyle{black point}=[point,fill=black,font=\color{white}]
\tikzstyle{black copoint}=[copoint,fill=black,font=\color{white}]

\tikzstyle{tiny gray point}=[tinypoint,fill=gray!40!white]

\tikzstyle{diredge}=[->]
\tikzstyle{ddiredge}=[<->]
\tikzstyle{rdiredge}=[<-]
\tikzstyle{thickdiredge}=[->, very thick]
\tikzstyle{pointer edge}=[->,very thick,gray]
\tikzstyle{pointer edge part}=[very thick,gray]
\tikzstyle{dashed edge}=[dashed]
\tikzstyle{thick dashed edge}=[very thick,dashed]
\tikzstyle{thick gray dashed edge}=[thick dashed edge,gray!40]
\tikzstyle{thick map edge}=[very thick,|->]

\makeatletter
\newcommand{\boxshape}[3]{%
\pgfdeclareshape{#1}{
\inheritsavedanchors[from=rectangle] % this is nearly a rectangle
\inheritanchorborder[from=rectangle]
\inheritanchor[from=rectangle]{center}
\inheritanchor[from=rectangle]{north}
\inheritanchor[from=rectangle]{south}
\inheritanchor[from=rectangle]{west}
\inheritanchor[from=rectangle]{east}
\backgroundpath{% this is new
\southwest \pgf@xa=\pgf@x \pgf@ya=\pgf@y
\northeast \pgf@xb=\pgf@x \pgf@yb=\pgf@y

\@tempdima=#2
\@tempdimb=#3

\pgfpathmoveto{\pgfpoint{\pgf@xa - 5pt + \@tempdima}{\pgf@ya}}
\pgfpathlineto{\pgfpoint{\pgf@xa - 5pt - \@tempdima}{\pgf@yb}}
\pgfpathlineto{\pgfpoint{\pgf@xb + 5pt + \@tempdimb}{\pgf@yb}}
\pgfpathlineto{\pgfpoint{\pgf@xb + 5pt - \@tempdimb}{\pgf@ya}}
\pgfpathlineto{\pgfpoint{\pgf@xa - 5pt + \@tempdima}{\pgf@ya}}
\pgfpathclose
}
}}

\boxshape{NEbox}{0pt}{5pt}
\boxshape{SEbox}{0pt}{-5pt}
\boxshape{NWbox}{5pt}{0pt}
\boxshape{SWbox}{-5pt}{0pt}
\boxshape{EBox}{-3pt}{3pt}
\boxshape{WBox}{3pt}{-3pt}
\makeatother

\tikzstyle{cloud}=[shape=cloud,draw,minimum width=1.5cm,minimum height=1.5cm]

\tikzstyle{map}=[draw,shape=NEbox,inner sep=2pt,minimum height=6mm,fill=white]
\tikzstyle{dashedmap}=[draw,dashed,shape=NEbox,inner sep=2pt,minimum height=6mm,fill=white]
\tikzstyle{mapdag}=[draw,shape=SEbox,inner sep=2pt,minimum height=6mm,fill=white]
\tikzstyle{mapadj}=[draw,shape=SEbox,inner sep=2pt,minimum height=6mm,fill=white]
\tikzstyle{maptrans}=[draw,shape=SWbox,inner sep=2pt,minimum height=6mm,fill=white]
\tikzstyle{mapconj}=[draw,shape=NWbox,inner sep=2pt,minimum height=6mm,fill=white]

\tikzstyle{medium map}=[draw,shape=NEbox,inner sep=2pt,minimum height=6mm,fill=white,minimum width=7mm]
\tikzstyle{medium map dag}=[draw,shape=SEbox,inner sep=2pt,minimum height=6mm,fill=white,minimum width=7mm]
\tikzstyle{medium map adj}=[draw,shape=SEbox,inner sep=2pt,minimum height=6mm,fill=white,minimum width=7mm]
\tikzstyle{medium map trans}=[draw,shape=SWbox,inner sep=2pt,minimum height=6mm,fill=white,minimum width=7mm]
\tikzstyle{medium map conj}=[draw,shape=NWbox,inner sep=2pt,minimum height=6mm,fill=white,minimum width=7mm]
\tikzstyle{semilarge map}=[draw,shape=NEbox,inner sep=2pt,minimum height=6mm,fill=white,minimum width=9.5mm]
\tikzstyle{semilarge map trans}=[draw,shape=SWbox,inner sep=2pt,minimum height=6mm,fill=white,minimum width=9.5mm]
\tikzstyle{semilarge map adj}=[draw,shape=SEbox,inner sep=2pt,minimum height=6mm,fill=white,minimum width=9.5mm]
\tikzstyle{semilarge map dag}=[draw,shape=SEbox,inner sep=2pt,minimum height=6mm,fill=white,minimum width=9.5mm]
\tikzstyle{semilarge map conj}=[draw,shape=NWbox,inner sep=2pt,minimum height=6mm,fill=white,minimum width=9.5mm]
\tikzstyle{large map}=[draw,shape=NEbox,inner sep=2pt,minimum height=6mm,fill=white,minimum width=12mm]
\tikzstyle{large map conj}=[draw,shape=NWbox,inner sep=2pt,minimum height=6mm,fill=white,minimum width=12mm]
\tikzstyle{very large map}=[draw,shape=NEbox,inner sep=2pt,minimum height=6mm,fill=white,minimum width=17mm]

\tikzstyle{medium dmap}=[draw,doubled,shape=NEbox,inner sep=2pt,minimum height=6mm,fill=white,minimum width=7mm]
\tikzstyle{medium dmap dag}=[draw,doubled,shape=SEbox,inner sep=2pt,minimum height=6mm,fill=white,minimum width=7mm]
\tikzstyle{medium dmap adj}=[draw,doubled,shape=SEbox,inner sep=2pt,minimum height=6mm,fill=white,minimum width=7mm]
\tikzstyle{medium dmap trans}=[draw,doubled,shape=SWbox,inner sep=2pt,minimum height=6mm,fill=white,minimum width=7mm]
\tikzstyle{medium dmap conj}=[draw,doubled,shape=NWbox,inner sep=2pt,minimum height=6mm,fill=white,minimum width=7mm]
\tikzstyle{semilarge dmap}=[draw,doubled,shape=NEbox,inner sep=2pt,minimum height=6mm,fill=white,minimum width=9.5mm]
\tikzstyle{semilarge dmap trans}=[draw,doubled,shape=SWbox,inner sep=2pt,minimum height=6mm,fill=white,minimum width=9.5mm]
\tikzstyle{semilarge dmap adj}=[draw,doubled,shape=SEbox,inner sep=2pt,minimum height=6mm,fill=white,minimum width=9.5mm]
\tikzstyle{semilarge dmap dag}=[draw,doubled,shape=SEbox,inner sep=2pt,minimum height=6mm,fill=white,minimum width=9.5mm]
\tikzstyle{semilarge dmap conj}=[draw,doubled,shape=NWbox,inner sep=2pt,minimum height=6mm,fill=white,minimum width=9.5mm]
\tikzstyle{large dmap}=[draw,doubled,shape=NEbox,inner sep=2pt,minimum height=6mm,fill=white,minimum width=12mm]
\tikzstyle{large dmap conj}=[draw,doubled,shape=NWbox,inner sep=2pt,minimum height=6mm,fill=white,minimum width=12mm]
\tikzstyle{large dmap trans}=[draw,doubled,shape=SWbox,inner sep=2pt,minimum height=6mm,fill=white,minimum width=12mm]
\tikzstyle{large dmap adj}=[draw,doubled,shape=SEbox,inner sep=2pt,minimum height=6mm,fill=white,minimum width=12mm]
\tikzstyle{large dmap dag}=[draw,doubled,shape=SEbox,inner sep=2pt,minimum height=6mm,fill=white,minimum width=12mm]
\tikzstyle{very large dmap}=[draw,doubled,shape=NEbox,inner sep=2pt,minimum height=6mm,fill=white,minimum width=19.5mm]

\tikzstyle{muxbox}=[draw,shape=rectangle,minimum height=3mm,minimum width=3mm,fill=white]
\tikzstyle{dmuxbox}=[muxbox,doubled]

\tikzstyle{box}=[draw,shape=rectangle,inner sep=2pt,minimum height=6mm,minimum width=6mm,fill=white]
\tikzstyle{dbox}=[draw,doubled,shape=rectangle,inner sep=2pt,minimum height=6mm,minimum width=6mm,fill=white]
\tikzstyle{dmap}=[draw,doubled,shape=NEbox,inner sep=2pt,minimum height=6mm,fill=white]
\tikzstyle{dmapdag}=[draw,doubled,shape=SEbox,inner sep=2pt,minimum height=6mm,fill=white]
\tikzstyle{dmapadj}=[draw,doubled,shape=SEbox,inner sep=2pt,minimum height=6mm,fill=white]
\tikzstyle{dmaptrans}=[draw,doubled,shape=SWbox,inner sep=2pt,minimum height=6mm,fill=white]
\tikzstyle{dmapconj}=[draw,doubled,shape=NWbox,inner sep=2pt,minimum height=6mm,fill=white]

\tikzstyle{ddmap}=[draw,doubled,dashed,shape=NEbox,inner sep=2pt,minimum height=6mm,fill=white]
\tikzstyle{ddmapdag}=[draw,doubled,dashed,shape=SEbox,inner sep=2pt,minimum height=6mm,fill=white]
\tikzstyle{ddmapadj}=[draw,doubled,dashed,shape=SEbox,inner sep=2pt,minimum height=6mm,fill=white]
\tikzstyle{ddmaptrans}=[draw,doubled,dashed,shape=SWbox,inner sep=2pt,minimum height=6mm,fill=white]
\tikzstyle{ddmapconj}=[draw,doubled,dashed,shape=NWbox,inner sep=2pt,minimum height=6mm,fill=white]

\boxshape{sNEbox}{0pt}{3pt}
\boxshape{sSEbox}{0pt}{-3pt}
\boxshape{sNWbox}{3pt}{0pt}
\boxshape{sSWbox}{-3pt}{0pt}
\tikzstyle{smap}=[draw,shape=sNEbox,fill=white]
\tikzstyle{smapdag}=[draw,shape=sSEbox,fill=white]
\tikzstyle{smapadj}=[draw,shape=sSEbox,fill=white]
\tikzstyle{smaptrans}=[draw,shape=sSWbox,fill=white]
\tikzstyle{smapconj}=[draw,shape=sNWbox,fill=white]

\tikzstyle{dsmap}=[draw,dashed,shape=sNEbox,fill=white]
\tikzstyle{dsmapdag}=[draw,dashed,shape=sSEbox,fill=white]
\tikzstyle{dsmaptrans}=[draw,dashed,shape=sSWbox,fill=white]
\tikzstyle{dsmapconj}=[draw,dashed,shape=sNWbox,fill=white]

\boxshape{mNEbox}{0pt}{10pt}
\boxshape{mSEbox}{0pt}{-10pt}
\boxshape{mNWbox}{10pt}{0pt}
\boxshape{mSWbox}{-10pt}{0pt}
\tikzstyle{mmap}=[draw,shape=mNEbox]
\tikzstyle{mmapdag}=[draw,shape=mSEbox]
\tikzstyle{mmaptrans}=[draw,shape=mSWbox]
\tikzstyle{mmapconj}=[draw,shape=mNWbox]

\tikzstyle{mmapgray}=[draw,fill=gray!40!white,shape=mNEbox]
\tikzstyle{smapgray}=[draw,fill=gray!40!white,shape=sNEbox]

\makeatletter

\pgfdeclareshape{cornerpoint}{
\inheritsavedanchors[from=rectangle] % this is nearly a rectangle
\inheritanchorborder[from=rectangle]
\inheritanchor[from=rectangle]{center}
\inheritanchor[from=rectangle]{north}
\inheritanchor[from=rectangle]{south}
\inheritanchor[from=rectangle]{west}
\inheritanchor[from=rectangle]{east}
\backgroundpath{% this is new
\southwest \pgf@xa=\pgf@x \pgf@ya=\pgf@y
\northeast \pgf@xb=\pgf@x \pgf@yb=\pgf@y

\pgfmathsetmacro{\pgf@shorten@left}{\pgfkeysvalueof{/tikz/shorten left}}
\pgfmathsetmacro{\pgf@shorten@right}{\pgfkeysvalueof{/tikz/shorten right}}

\pgfpathmoveto{\pgfpoint{0.5 * (\pgf@xa + \pgf@xb)}{\pgf@ya - 5pt}}
\pgfpathlineto{\pgfpoint{\pgf@xa - 8pt + \pgf@shorten@left}{\pgf@yb - 1.5 * \pgf@shorten@left}}
\pgfpathlineto{\pgfpoint{\pgf@xa - 8pt + \pgf@shorten@left}{\pgf@yb}}
\pgfpathlineto{\pgfpoint{\pgf@xb + 8pt - \pgf@shorten@right}{\pgf@yb}}
\pgfpathlineto{\pgfpoint{\pgf@xb + 8pt - \pgf@shorten@right}{\pgf@yb - 1.5 * \pgf@shorten@right}}
\pgfpathclose
}
}

\pgfdeclareshape{cornercopoint}{
\inheritsavedanchors[from=rectangle] % this is nearly a rectangle
\inheritanchorborder[from=rectangle]
\inheritanchor[from=rectangle]{center}
\inheritanchor[from=rectangle]{north}
\inheritanchor[from=rectangle]{south}
\inheritanchor[from=rectangle]{west}
\inheritanchor[from=rectangle]{east}
\backgroundpath{% this is new
\southwest \pgf@xa=\pgf@x \pgf@ya=\pgf@y
\northeast \pgf@xb=\pgf@x \pgf@yb=\pgf@y

\pgfmathsetmacro{\pgf@shorten@left}{\pgfkeysvalueof{/tikz/shorten left}}
\pgfmathsetmacro{\pgf@shorten@right}{\pgfkeysvalueof{/tikz/shorten right}}

\pgfpathmoveto{\pgfpoint{0.5 * (\pgf@xa + \pgf@xb)}{\pgf@yb + 5pt}}
\pgfpathlineto{\pgfpoint{\pgf@xa - 8pt + \pgf@shorten@left}{\pgf@ya + 1.5 * \pgf@shorten@left}}
\pgfpathlineto{\pgfpoint{\pgf@xa - 8pt + \pgf@shorten@left}{\pgf@ya}}
\pgfpathlineto{\pgfpoint{\pgf@xb + 8pt - \pgf@shorten@right}{\pgf@ya}}
\pgfpathlineto{\pgfpoint{\pgf@xb + 8pt - \pgf@shorten@right}{\pgf@ya + 1.5 * \pgf@shorten@right}}
\pgfpathclose
}
}

\makeatother

\pgfkeyssetvalue{/tikz/shorten left}{0pt}
\pgfkeyssetvalue{/tikz/shorten right}{0pt}

\tikzstyle{kpoint common}=[draw,fill=white,inner sep=1pt,minimum height=4mm]
\tikzstyle{kpoint sc}=[shape=cornerpoint,kpoint common]
\tikzstyle{kpoint adjoint sc}=[shape=cornercopoint,kpoint common]
\tikzstyle{kpoint}=[shape=cornerpoint,shorten left=5pt,kpoint common]
\tikzstyle{kpoint adjoint}=[shape=cornercopoint,shorten left=5pt,kpoint common]
\tikzstyle{kpoint conjugate}=[shape=cornerpoint,shorten right=5pt,kpoint common]
\tikzstyle{kpoint transpose}=[shape=cornercopoint,shorten right=5pt,kpoint common]
\tikzstyle{kpoint symm}=[shape=cornerpoint,shorten left=5pt,shorten right=5pt,kpoint common]

\tikzstyle{wide kpoint sc}=[shape=cornerpoint,kpoint common, minimum width=1 cm]
\tikzstyle{wide kpointdag sc}=[shape=cornercopoint,kpoint common, minimum width=1 cm]

\tikzstyle{black kpoint}=[shape=cornerpoint,shorten left=5pt,kpoint common,fill=black,font=\color{white}]

\tikzstyle{black kpoint sm}=[shape=cornerpoint,shorten left=5pt,kpoint common,fill=black,font=\color{white},scale=0.75]

\tikzstyle{black kpoint adjoint}=[shape=cornercopoint,shorten left=5pt,kpoint common,fill=black,font=\color{white}]
\tikzstyle{black kpointadj}=[shape=cornercopoint,shorten left=5pt,kpoint common,fill=black,font=\color{white}]

\tikzstyle{black kpointadj sm}=[shape=cornercopoint,shorten left=5pt,kpoint common,fill=black,font=\color{white},scale=0.75]

\tikzstyle{black dkpoint}=[shape=cornerpoint,shorten left=5pt,kpoint common,fill=black, doubled,font=\color{white}]
\tikzstyle{black dkpoint adjoint}=[shape=cornercopoint,shorten left=5pt,kpoint common,fill=black, doubled,font=\color{white}]
\tikzstyle{black dkpointadj}=[shape=cornercopoint,shorten left=5pt,kpoint common,fill=black, doubled,font=\color{white}]

\tikzstyle{black dkpoint sm}=[shape=cornerpoint,shorten left=5pt,kpoint common,fill=black, doubled,font=\color{white},scale=0.75]
\tikzstyle{black dkpointadj sm}=[shape=cornercopoint,shorten left=5pt,kpoint common,fill=black, doubled,font=\color{white},scale=0.75]

\tikzstyle{kpointdag}=[kpoint adjoint]
\tikzstyle{kpointadj}=[kpoint adjoint]
\tikzstyle{kpointconj}=[kpoint conjugate]
\tikzstyle{kpointtrans}=[kpoint transpose]

\tikzstyle{big kpoint}=[kpoint, minimum width=1.2 cm, minimum height=8mm, inner sep=4pt, text depth=3mm]

\tikzstyle{wide kpoint}=[kpoint, minimum width=1 cm, inner sep=2pt]%, text depth=-0.7 mm]
\tikzstyle{wide kpointdag}=[kpointdag, minimum width=1 cm, inner sep=2pt]%, text depth=0.7 mm]
\tikzstyle{wide kpointconj}=[kpointconj, minimum width=1 cm, inner sep=2pt]%, text depth=-0.7 mm]
\tikzstyle{wide kpointtrans}=[kpointtrans, minimum width=1 cm, inner sep=2pt]%, text depth=0.7 mm]

\tikzstyle{wider kpoint}=[kpoint, minimum width=1.25 cm, inner sep=2pt]%, text depth=-0.7 mm]
\tikzstyle{wider kpointdag}=[kpointdag, minimum width=1.25 cm, inner sep=2pt]%, text depth=0.7 mm]
\tikzstyle{wider kpointconj}=[kpointconj, minimum width=1.25 cm, inner sep=2pt]%, text depth=-0.7 mm]
\tikzstyle{wider kpointtrans}=[kpointtrans, minimum width=1.25 cm, inner sep=2pt]%, text depth=0.7 mm]

\tikzstyle{gray kpoint}=[kpoint,fill=gray!50!white]
\tikzstyle{gray kpointdag}=[kpointdag,fill=gray!50!white]
\tikzstyle{gray kpointadj}=[kpointadj,fill=gray!50!white]
\tikzstyle{gray kpointconj}=[kpointconj,fill=gray!50!white]
\tikzstyle{gray kpointtrans}=[kpointtrans,fill=gray!50!white]

\tikzstyle{gray dkpoint}=[kpoint,fill=gray!50!white,doubled]
\tikzstyle{gray dkpointdag}=[kpointdag,fill=gray!50!white,doubled]
\tikzstyle{gray dkpointadj}=[kpointadj,fill=gray!50!white,doubled]
\tikzstyle{gray dkpointconj}=[kpointconj,fill=gray!50!white,doubled]
\tikzstyle{gray dkpointtrans}=[kpointtrans,fill=gray!50!white,doubled]

\tikzstyle{white label}=[draw,fill=white,rectangle,inner sep=0.7 mm]
\tikzstyle{gray label}=[draw,fill=gray!50!white,rectangle,inner sep=0.7 mm]
\tikzstyle{black label}=[draw,fill=black,rectangle,inner sep=0.7 mm]

\tikzstyle{dkpoint}=[kpoint,doubled]
\tikzstyle{wide dkpoint}=[wide kpoint,doubled]
\tikzstyle{dkpointdag}=[kpoint adjoint,doubled]
\tikzstyle{wide dkpointdag}=[wide kpointdag,doubled]
\tikzstyle{dkcopoint}=[kpoint adjoint,doubled]
\tikzstyle{dkpointadj}=[kpoint adjoint,doubled]
\tikzstyle{dkpointconj}=[kpoint conjugate,doubled]
\tikzstyle{dkpointtrans}=[kpoint transpose,doubled]

\tikzstyle{kscalar}=[kpoint common, shape=EBox, inner xsep=-1pt, inner ysep=3pt,font=\small]
\tikzstyle{kscalarconj}=[kpoint common, shape=WBox, inner xsep=-1pt, inner ysep=3pt,font=\small]

\tikzstyle{spekpoint}=[kpoint sc,minimum height=5mm,inner sep=3pt]
\tikzstyle{spekcopoint}=[kpoint adjoint sc,minimum height=5mm,inner sep=3pt]

\tikzstyle{dspekpoint}=[spekpoint,doubled]
\tikzstyle{dspekcopoint}=[spekcopoint,doubled]

 \tikzstyle{upground}=[circuit ee IEC,thick,ground,rotate=90,scale=2.5]
 \tikzstyle{downground}=[circuit ee IEC,thick,ground,rotate=-90,scale=2.5]
 \tikzstyle{bigground}=[regular polygon,regular polygon sides=3,draw=gray,scale=0.50,inner sep=-0.5pt,minimum width=10mm,fill=gray]
\tikzstyle{arrs}=[-latex,font=\small,auto]
\tikzstyle{arrow plain}=[arrs]
\tikzstyle{arrow dashed}=[dashed,arrs]
\tikzstyle{arrow bold}=[very thick,arrs]
\tikzstyle{arrow hide}=[draw=white!0,-]
\tikzstyle{arrow reverse}=[latex-]
\tikzstyle{cdnode}=[]

\tikzstyle{discarding}=[fill=white, draw=black, shape=circle, style=upground]
\tikzstyle{smalldiscarding}=[fill=white, draw=black, style=upground, scale=0.5]
\tikzstyle{backdiscard}=[fill=white, draw=black, shape=circle, style=downground, scale=0.5]
\tikzstyle{smallbackdiscard}=[fill=white, draw=black, shape=circle, style=downground, scale=0.5]
\tikzstyle{state}=[fill=white, draw=black, style=triang, tikzit shape=rectangle]
\tikzstyle{kstate}=[fill=white, draw=black, style=kpoint, tikzit shape=rectangle]
\tikzstyle{kstateconj}=[fill=white, draw=black, style=kpoint conjugate, tikzit shape=rectangle]
\tikzstyle{kstateBIG}=[fill=white, draw=black, style=big kpoint, tikzit shape=rectangle]
\tikzstyle{effect}=[fill=white, draw=black, style=triangdag]
\tikzstyle{keffect}=[fill=white, draw=black, style=kpoint adjoint]
\tikzstyle{keffectconj}=[fill=white, draw=black, style=kpoint transpose]
\tikzstyle{morphdag}=[style=mapdag]
\tikzstyle{morph}=[style=hadamard]
\tikzstyle{WIDEmorph}=[style=hadamard, minimum width=14mm]
\tikzstyle{morphtrans}=[style=maptrans]
\tikzstyle{morphconj}=[style=mapconj]
\tikzstyle{CPMmorph}=[style=dmap]
\tikzstyle{CPMmorphconj}=[style=dmapconj]
\tikzstyle{CPMmorphdag}=[style=dmapdag]
\tikzstyle{CPMmorphtrans}=[style=dmaptrans]
\tikzstyle{CPMstate}=[fill=white, draw=black, style=triang, doubled]
\tikzstyle{CPMstateBIG}=[fill=white, draw=black, style={triang_lesssep}, doubled]
\tikzstyle{CPMkstate}=[fill=white, draw=black, style=kpoint, tikzit shape=rectangle, doubled]
\tikzstyle{CPMkstateconj}=[fill=white, draw=black, style=kpoint conjugate, tikzit shape=rectangle, doubled]
\tikzstyle{CPMkstateBIG}=[fill=white, draw=black, style=big kpoint, tikzit shape=rectangle, doubled]
\tikzstyle{CPMkeffect}=[fill=white, draw=black, style=kpoint adjoint, doubled]
\tikzstyle{CPMkeffectconj}=[fill=white, draw=black, style=kpoint transpose, doubled]
\tikzstyle{UHfB}=[fill=white, draw=black, style=triangdag, doubled, inner sep=-2pt]
\tikzstyle{leak}=[style=tinypoint, regular polygon rotate=-90]
\tikzstyle{leakfill}=[style=tinypoint, regular polygon rotate=-90, fill=black]
\tikzstyle{Z}=[style=dot, fill=green]
\tikzstyle{X}=[style=dot, fill=red]
\tikzstyle{black_dot}=[style=dot, fill=black]
\tikzstyle{white_dot}=[style=dot, fill=white]
\tikzstyle{qblack_dot}=[style=ddot, fill=black]
\tikzstyle{qwhite_dot}=[style=ddot, fill=white]
\tikzstyle{whitephase}=[style=wphase dot, fill=white]
\tikzstyle{qredphase}=[style=phase dot, fill=red]
\tikzstyle{qgreenphase}=[style=phase dot, fill=green]
\tikzstyle{had}=[style=hadamard, doubled]
\tikzstyle{box}=[style=hadamard]
\tikzstyle{classhad}=[style=hadamard]
\tikzstyle{antipode}=[style=anti]

\tikzstyle{dottededge}=[-, dotted]
\tikzstyle{double_edge}=[-, double]
\tikzstyle{double_dot}=[-, double, dotted]
\tikzstyle{double_bold}=[very thick, double]
\tikzstyle{double edge}=[-, very thick]
\tikzstyle{new edge style 0}=[<-]
\tikzstyle{new edge style 1}=[-, draw={rgb,255: red,255; green,204; blue,204}, fill={rgb,255: red,255; green,204; blue,204}]
\tikzstyle{new edge style 2}=[-, draw={rgb,255: red,14; green,188; blue,83}]
\tikzstyle{new edge style 3}=[<-, draw={rgb,255: red,255; green,204; blue,204}]
\tikzstyle{new edge style 4}=[<-, draw={rgb,255: red,0; green,128; blue,128}]
\tikzstyle{new edge style 5}=[-, draw={rgb,255: red,214; green,110; blue,62}]
\tikzstyle{new edge style 6}=[-, draw={rgb,255: red,174; green,20; blue,174}]

\usetikzlibrary{calc, positioning, shapes.geometric}
\usetikzlibrary{
	arrows,
	shapes,
	decorations,
	intersections,
	backgrounds,
	positioning,
	circuits.ee.IEC
	}

\newcommand{\tikzfigscale}[2]{\scalebox{#1}{\tikzfig{#2}}}

\def\be{\begin{equation}}
\def\ee{\end{equation}}
\def\ba{\begin{align}}
\def\ea{\end{align}}

\usepackage{bbm}

\newcommand{\lalg}[1]{\mathfrak{#1}}  % Lie algebra

\renewcommand{\sl}{\lalg{sl}}

\newtheorem{definition}{Definition}
\newtheorem{theorem}{Theorem}

\newtheorem{lemma}{Lemma}

\newtheorem{example}{Example}

\tikzstyle{every picture}=[baseline=-0.25em,shorten <=-0.1pt]
\tikzstyle{dotpic}=[scale=0.5]
\tikzstyle{braceedge}=[decorate,decoration={brace,amplitude=1mm,raise=-1mm}]
\tikzstyle{dot}=[inner sep=0.7mm,minimum width=0pt,minimum height=0pt,fill=black,draw=black,shape=circle]
\tikzstyle{small dot}=[inner sep=0.1mm,minimum width=0pt,minimum height=0pt,fill=black,draw=black,shape=circle]
\tikzstyle{black dot}=[dot]
\tikzstyle{white dot}=[dot,fill=white]
\tikzstyle{gray dot}=[dot,fill=gray!40!white]
\tikzstyle{alt white dot}=[white dot,label={[xshift=3mm,yshift=-0.05mm,font=\tiny]left:$*$}]
\tikzstyle{alt gray dot}=[gray dot,label={[xshift=3mm,yshift=-0.05mm,font=\tiny]left:$*$}]
\tikzstyle{white norm}=[rectangle,fill=white,draw=black,minimum height=2mm,minimum width=2mm,inner sep=0pt,font=\small]
\tikzstyle{gray norm}=[white norm,fill=gray!40!white]
\tikzstyle{square box}=[rectangle,fill=white,draw=black,minimum height=5mm,minimum width=5mm,font=\small]
\tikzstyle{square gray box}=[rectangle,fill=gray!30,draw=black,minimum height=6mm,minimum width=6mm]
\tikzstyle{diredge}=[->]
\tikzstyle{rdiredge}=[<-]
\tikzstyle{dashed edge}=[dashed]
\tikzstyle{cross}=[preaction={draw=white, -, line width=3pt}]

\newcommand{\dotdualmult}[1]{%
\!\begin{tikzpicture}[dotpic]
    \node [style=white dot] (0) at (0, 0.3) {};
    \node [style=none] (1) at (-0.5, -0.4) {};
    \node [style=none] (2) at (0.5, -0.4) {};
    \node [style=none] (3) at (0, 0.8) {};
    \draw [style=diredge] (3.center) to (0);
    \draw [style=diredge, in=15, out=-30, looseness=1.50] (0) to (1.center);
    \draw [style=diredge, in=165, out=-150, looseness=1.50] (0) to (2.center);
\end{tikzpicture}\!}

\newcommand{\dotconorm}[1]{%
\,\begin{tikzpicture}[dotpic,yshift=0.4mm]
    \node [style=none] (0) at (0, -0.4) {};
    \node [style=white norm] (1) at (0, 0.1) {};
    \node [style=none] (2) at (0, 0.5) {};
    \draw [style=diredge] (1) to (0.center);
    \draw (2.center) to (1);
\end{tikzpicture}\,}

\newcommand{\astfootnote}[1]{
\let\oldthefootnote=\thefootnote
\setcounter{footnote}{0}
\renewcommand{\thefootnote}{\fnsymbol{footnote}}
\footnote{#1}
\let\thefootnote=\oldthefootnote
}

\title{Polycategorical Constructions for Unitary Supermaps of Arbitrary Dimension}

\author{Matt Wilson}
\email{matthew.wilson@cs.ox.ac.uk}
\affiliation{Quantum Group, Department of Computer Science, University of Oxford}
\affiliation{HKU-Oxford Joint Laboratory for Quantum Information and Computation}
\affiliation{Programming Principles Logic and Verification Group,
University College London
London, UK}
\affiliation{Universit{e} Paris-Saclay, CNRS, ENS Paris-Saclay, Inria, CentraleSup\'{e}lec, Laboratoire M{e}thodes Formelles}

\author{Giulio Chiribella}
\email{giulio.chiribella@cs.ox.ac.uk}
\affiliation{Quantum Group, Department of Computer Science, University of Oxford}
\affiliation{HKU-Oxford Joint Laboratory for Quantum Information and Computation}
\affiliation{QICI Quantum Information and Computation Initiative, Department of Computer Science}
\affiliation{Perimeter Institute for Theoretical Physics, 31 Caroline Street North, Waterloo, Ontario, Canada}

\begin{document} \emergencystretch 3em

\maketitle

\begin{abstract}
We provide a construction for holes into which morphisms of abstract symmetric monoidal categories can be inserted, termed the \textit{polyslot} construction $\mathbf{pslot}[\mathbf{C}]$, and identify a sub-class $\mathbf{srep}[\mathbf{C}]$ of polyslots which are \textit{single-party representable}. These constructions strengthen a previously introduced notion of locally-applicable transformation used to characterize quantum supermaps in a way that is sufficient to reconstruct unitary supermaps directly from the monoidal structure of the category of unitaries. Both constructions furthermore freely reconstruct the enriched polycategorical semantics for quantum supermaps which allows to compose supermaps in sequence and in parallel whilst forbidding the creation of time-loops. By doing so supermaps and their polycategorical semantics are generalized to infinite dimensions, in such a way as to include canonical examples such as the quantum switch. Beyond specific applications to quantum-relevant categories, a general class of categorical monoidal categories termed path-contraction groupoids are defined on which the $\mathbf{srep}[\mathbf{C}]$ and $\mathbf{pslot}[\mathbf{C}]$ constructions are shown to coincide.
\end{abstract}

\section{Introduction}
A key concept in a variety of scientific and mathematical disciplines is the specification of two classes of data, a collection of systems, and a specification processes which act upon those systems. A common emergent theme within some such fields has been the development of the concept of a hole into-which a process could be inserted, such instances can be seen within the study of higher-order quantum computation \cite{Chiribella2008QuantumArchitectureb, Chiribella2008TransformingSupermaps, Chiribella2009QuantumStructure, Chiribella2010NormalOperations, Chiribella2009TheoreticalNetworks}, quantum causality \cite{Chiribella2008QuantumArchitectureb, Chiribella2008TransformingSupermaps, Chiribella2009QuantumStructure, Chiribella2010NormalOperations, Chiribella2009TheoreticalNetworks, Oreshkov2012QuantumOrder}, bidirectional programming \cite{Riley2018CategoriesOptics, Boisseau2020StringOptics, Boisseau2022CorneringOptics}, game-theory \cite{Hedges2017CoherenceGames, Hedges2019TheTheory, Bolt2019BayesianGames, Ghani2016CompositionalTheory}, machine learning \cite{Cruttwell2021CategoricalLearning}, open systems dynamics \cite{vanderMeulen2020AutomaticModels}, quantum open systems dynamics \cite{Pollock2015Non-MarkovianCharacterisation}, categorical cybernetics \cite{Hedges2024YogaContexts}, and even financial trading \cite{Genovese2021EscrowsOptics}. 

A natural primitive notion of diagram-with-hole for an arbitrary symmetric monoidal category can be given by taking a circuit diagram term, and puncturing a series of holes into it: \[   \tikzfigscale{0.75}{figs/introcomb}  \]
Such diagrams have been studied in quantum theory under the name of quantum-combs \cite{Chiribella2008QuantumArchitectureb}, and in bidirectional programming as profunctor optics \cite{Riley2018CategoriesOptics, Roman2020CombFeedback, Roman2020OpenCalculus} with the two approaches connected in the unitary case by \cite{Hefford2022CoendCombs}. However, in quantum contexts considerable attention is given to a generalisation of the above picture to black-box holes called quantum supermaps \cite{Chiribella2008TransformingSupermaps}, which are not assumed to be expressed as circuit-diagrams of the above form. The canonical example of such a supermap is the quantum switch \cite{Chiribella2013QuantumStructure, Chiribella2009QuantumStructure} which represents a quantum superposition of two possible diagrams with open holes \[   \tikzfigscale{1.3}{figs/introswitch} \quad = \quad \tikzfigscale{1.3}{figs/introswitch0} \quad + \quad \tikzfigscale{1.3}{figs/introswitch1}.  \]
The concept of a black-box hole, which processes may be plugged into, is at the intuitive level easy enough to imagine, and yet, it has been unclear how to generalize quantum supermaps appropriately to arbitrary operational probabilistic theories \cite{Chiribella2009ProbabilisticPurification, Chiribella2016QuantumPrinciples}, including to infinite-dimensional quantum theory. A proposal of \cite{Chiribella2010NormalOperations} refers only to single inputs, and it is unclear whether the proposal of \cite{Giacomini2015IndefiniteSystems} produces maps that can be suitably extended to be applied on part of any bipartite process. A proposal of \cite{Wilson2022QuantumLocality} is to use $*\mathbf{Hilb}$ \cite{Gogioso2019QuantumMechanics, Gogioso2018TowardsMechanics} which produces fairly well-behaved results but however requires understanding of the use of non-standard analysis or $2$-category theory and as currently defined is only appropriate for the unitary (non-mixed) setting. The issue, in short, is that whilst the spirit of the definition of quantum supermaps is intended to be abstract and black-box, in practice the definition of supermap on a physical theory requires knowledge of mathematical structure beyond the circuit-theoretic structure of that theory, such as the existence of an appropriate raw-material category into which the quantum channels embed \cite{Kissinger2019AStructure, Bisio2019TheoreticalTheory, Simmons2022Higher-orderBV-logic, ApadulaNo-signallingStructure}. 

This article is written with the aim of suggesting appropriate definitions for supermaps that require \textit{only} the circuit-theoretic structure of the categories they act on in their definition. An exploration of the available definitions of supermaps in general symmetric monoidal categories is expected to have two main applications, first, a satisfactory generalization to infinite dimensions would allow to make a connection between the supermap program and the program of unification of quantum theory with gravitational physics, where quantum causal structures such as those present in the quantum switch are predicted by some to play a key conceptual role \cite{Hardy2006TowardsStructure, Oreshkov2012QuantumOrder}. Beyond applications to quantum gravity, a principled definition of black-box hole, and exploration of the landscape of possible definitions may be of use to those other fields in which circuits with open holes are currently studied.

In previous work, a definition of locally-applicable transformation was proposed for modeling black-box holes in general symmetric monoidal categories, and shown to recover the quantum supermaps when applied to the symmetric monoidal category of quantum channels \cite{Wilson2022QuantumLocality}. The key principle was to capture the following expected behaviour of a hole, the possibility to apply to part of any bipartite process whilst commuting with local actions on the environment: \[ \tikzfig{figs/informal_cp_4aa} \quad = \quad \tikzfig{figs/informal_cp_4ab}.  \]
Whilst in this work locally-applicable transformations on quantum channels were shown to be in one-to-one correspondence with quantum supermaps, there are two properties we desire for a construction of supermaps on arbitrary symmetric monodial categories which are not exhibited by the definition of a locally-applicable transformation.
\begin{itemize}
    \item First, we aim to find a construction on symmetric monoidal categories which when applied to the category $\mathbf{fU}$ of finite-dimensional unitary processes, recovers the unitary-preserving quantum supermaps.
    \item  Secondly, we aim to find a construction that allows to unambiguously give formal meaning to the following intuitive pictures that one would like to safely imagine when thinking about such holes: \[   \tikzfig{figs/informal_par}  \quad \quad \quad \quad  \tikzfig{figs/informal_poly1}   \] In short we desire a construction strong enough freely give a monoidal \cite{Lane1971CategoriesMathematician} and polycategorical \cite{szabo_polycats} semantics for holes in symmetric monoidal categories, capturing the heart of the linear distributivity of previous approaches to constructing categories of quantum supermaps \cite{Kissinger2019AStructure, Simmons2022Higher-orderBV-logic}.
\end{itemize}
In this paper, we provide two stronger constructions that satisfy these requirements. The first construction, termed the $\mathbf{srep}[\mathbf{C}]$ construction, reconstructs supermaps by assuming a powerful structural theorem, that as viewed by single parties, they act as combs \cite{Chiribella2008TransformingSupermaps}. By developing a second construction of a polycategory of \textit{polyslots} termed $\mathbf{pslot}[\mathbf{C}]$ we show that the decomposition of supermaps at the single-party level as combs is a consequence in unitary quantum theory of a strong-locality principle. This strong-locality can be interpreted as taking the bi-commutant of the family of combs, and so connects the definition of supermaps to the definition of subsystems as bi-commutant families of operations \cite{Gogioso2019ASpace}. In our first class of results, we show that the above constructions indeed return polycategories.  
\begin{theorem}
$\mathbf{pslot}[\mathbf{C}]$ and $\mathbf{srep}[\mathbf{C}]$ are symmetric polycategories
\end{theorem}
In our second class of results, we prove that in a broad class of categories termed ``path contraction groupoids" (which include all symmetric monoidal groupoids which are subcategories of compact closed categories) the above constructions coincide.
\begin{theorem}
Let $\mathbf{G}$ be a path-contraction groupoid, then $\mathbf{pslot}[\mathbf{G}]  =\mathbf{srep}[\mathbf{G}]$
\end{theorem}
As a corollary of this theorem, we find that either construction characterizes the finite dimensional quantum supermaps in both the mixed and unitary cases.
\begin{theorem}Polyslots generalize quantum supermaps on the quantum channels and on the unitaries to arbitrary symmetric monoidal categories. Formally, there is an equivalence \[ \mathbf{pslot}[\mathbf{fU}] \cong \mathbf{uQS} \] of polycategories where $\mathbf{fU}$ is the catetegory of unitaries and $\mathbf{uQS}$ the polycategory of unitary-preserving quantum supermaps along with an equivalence \[ \mathbf{pslot}[\mathbf{fQC}] \cong \mathbf{QS} \] of polycategories where $\mathbf{fQC}$ is the category of finite-dimensional quantum channels and $\mathbf{QS}$ the polycategory of quantum supermaps.
\end{theorem}
In applications to infinite-dimensional unitary quantum theory, we further find that $\mathbf{pslot}[\mathbf{sepU}]$ and equivalently $\mathbf{srep}[\mathbf{sepU}]$ are always implementable by time-loops and unitaries, where $\mathbf{sepU}$ is the category of unitary linear maps between separable Hilbert spaces. Whilst $\mathbf{pslot}[\mathbf{sepU}]$ is strong enough to enforce a polycategorical semantics for infinite-dimensional unitary-preserving supermaps, we find that $\mathbf{pslot}[\mathbf{sepU}]$ is still flexible enough to include generalizations of motivating instances of quantum supermaps such as the quantum switch to infinite dimensions. The applications of this general black-box definition of hole to the growing number of scientific fields in which open diagrams are studied is left for future discovery, as is the extension of the construction to include the more elaborate and iterated type-systems developed for handling higher-order quantum theory in a series of recent works \cite{Bisio2019TheoreticalTheory, Perinotti2017CausalComputations, Kissinger2019AStructure, Simmons2022Higher-orderBV-logic, ApadulaNo-signallingStructure}.

\section{Preliminary Material}
Here we introduce the category-theoretic terms used throughout the paper, category theory is used here purely as an organizing language, and all calculations are written in a way that is aimed to be followable by those who are familiar only with string diagrams for compact closed categories. In general, we will adopt the convention of representing processes that are higher-order in white and processes which are lower order in pink, this choice has no formal significance and is made purely for readability.
\paragraph{Category Theory}
A \textit{category} \cite{Lane1971CategoriesMathematician} consists in a specification of objects $A,B,C, \dots$ and a specification of morphisms which act between them. Formally a category is equipped with, for each pair $A,B$ of objects a set $\mathbf{C}(A,B)$ terms the set of ``morphisms". A category furthermore is equipped with a composition function $\circ : \mathbf{C}(A,B) \times \mathbf{C}(B,C) \rightarrow \mathbf{C}(A,C)$ denoted $\circ$ for each $A,B,C$ such that $f \circ (g \circ h) = (f \circ g) \circ h$. Categories come with unit morphisms $id_{A}:A \rightarrow A$ for each object $A$ such that for each $f:A \rightarrow B$ then $f \circ id_A = f = id_B \circ f$. The defining conditions of a category can be conveniently absorbed into a graphical language which makes clear their suitability for representing processes between systems. An object $A$ of a category can always be represented by a wire, and a morphism $f:A \rightarrow B$ by a box with input wire $A$ and output wire $B$:  \[  \tikzfig{figs/category_2} . \]  Sequential composition $f \circ g$ is denoted \[  \tikzfig{figs/category_3},  \]
with associativity allowing for unambiguous interpretation of the diagram. The identity process can be represented by a wire \[  \tikzfig{figs/category_1} , \] so that again the defining equation $f \circ id_A = id_A$ is absorbed into the graphical notation. 
%As far as theories of information processing are concerned the sequential composition in the definitions data of a category is often interpreted as specification of a notion of time-ordered application of processes.
We describe a morphism $f:A \rightarrow B$ as an \textit{isomorphism} if there exists $\bar{f}:B \rightarrow A$ such that $f \circ \bar{f} = id_{B}$ and $\bar{f} \circ f = id_{B}$. If every morphism of a category $\mathbf{C}$ is an isomorphism then $\mathbf{C}$ is termed a \textit{groupoid}.

\paragraph{Monoidal categories}
In the process-theoretic approach to physics \cite{Coecke2017PicturingReasoning, Coecke2010QuantumPicturalism, Coecke2006KindergartenNotes}, the primary object of study is that of a circuit-theory. Monoidal categories give an algebraic model for circuit theories in terms of sequential and parallel composition operations. Formally, a \textit{monoidal} category is a category $\mathbf{C}$ equipped with a functor $\otimes: \mathbf{C} \times \mathbf{C} \rightarrow \mathbf{C}$ which assigns to each pair $(A,B)$ of objects in $\mathbf{C}$ an object $A \otimes B$ again in $\mathbf{C}$ \[  \tikzfig{figs/category_4}.  \] Similarly to each pair $(f,g)$ or morphisms the functor $\otimes$ assigns a new morphism $f \otimes g$. Functorality of $\otimes$ also implies interchange laws such as $(f \otimes id) \circ (id \otimes g) = (id \otimes g) \circ (f \otimes id)$ which can be represented diagrammatically by box-sliding \[  \tikzfig{figs/category_int1} \quad = \quad  \tikzfig{figs/category_int2} .\] 

Beyond monoidal categories one can define those which are \textit{symmetric}, meaning that they are equipped with a braid $\beta_{B,A}:B \otimes A \rightarrow A \otimes B$ depicted graphically by \[  \tikzfig{figs/category_6}.  \]  when applied twice the condition of symmetry further requires that $\beta_{B,A} \circ \beta_{A,B} = id_{B \otimes A}$, which essentially entails that the spatial position of wires on the page is of relevance only as-so-far as it is useful for book-keeping. If a monoidal category is a groupoid we will term it a \textit{monoidal groupoid}.

\paragraph{Compact closed categories}
A \textit{compact closed} category is one in which arbitrary input and output wires can be plugged together. Formally a compact closed category is a symmetric monoidal category $\mathbf{C}$ equipped with for each object $A$ a ``dual" object $A^{*}$, a state $\cup:I \rightarrow A^{*} \otimes A$ and effect $\cup: A^{*} \otimes A \rightarrow I$ such that \[  \tikzfig{figs/compact_1}  \quad = \quad \tikzfig{figs/compact_2}.  \] often referred to as the snake equation. A key feature of the snake equation is that it equips a monoidal category with an equivalence between inputs and outputs, this is a practically useful graphical property that allows the representation of process/state duality and feedback in monoidal categories in an internalized way. 

\paragraph{Polycategories}
There are non-monoidal algebraic structures within which interchange and associativity laws can be specified. Polycategories \cite{szabo_polycats}, provide an instance of such structures relevant to this paper. A polycategory is given by specification of a class of atomic objects, and then morphisms are defined as going between lists of such atomic objects \[ f:\underline{A} \rightarrow \underline{B} \quad \quad \quad \quad \underline{A} = A_1 \dots A_n \quad , \quad \underline{B} = B_1 \dots B_m . \] Whilst monoidal structure allows to compose along many objects at once, poly-categorical structure allows to compose morphisms along individually specified objects. 
Morphisms of polycategories can be written just as they would be for monoidal categories \[  \tikzfigscale{0.7}{figs/poly_new0} \quad = \quad \tikzfigscale{0.7}{figs/poly_new2},  \] with composition denoted by \[  \tikzfigscale{0.7}{figs/polycomp1} \quad \circ_{M} \quad \tikzfigscale{0.7}{figs/polycomp2} \quad = \quad \tikzfigscale{0.7}{figs/polycomp3} . \]  For the diagrammatic representation to be sound this composition rule should satisfy a variety of conditions.
%there should be no difference in interpretation of the following pictures \[  \tikzfigscale{0.7}{figs/polycomp4} \quad = \quad \tikzfigscale{0.7}{figs/polycomp5} \] amongst others. 
Formally, following \cite{shulman20202chudialectica}, a \textit{polycategory} comes equipped with
%given $f:\underline{A} \rightarrow \underline{B} N \underline{C}  $ and $g:\underline{D} N \underline{E} \rightarrow \underline{F}$ one may construct the composition $g \circ_{N} f : \underline{D} \underline{A}  \underline{E} \rightarrow  \underline{B}  \underline{F}  \underline{C}$ along $N$. 
\begin{itemize}
\item A functorial action by permutations, meaning for each pair of lists $\underline{A}, \underline{B}$ of cardinalities $n,m$ respectively and for each morphism $f:\underline{A} \rightarrow \underline{B}$ and pair of permutations $\sigma:[n] \rightarrow [n]$ and $\rho: [m] \rightarrow [m]$ a new morphism denoted $\rho (f) \sigma$ such that $\rho' (\rho (f) \sigma) \sigma' = (\rho \circ \rho')(f)(\sigma' \circ \sigma)$.
\item For each pair $f:\underline{A} \rightarrow \underline{B} X \underline{C}, g: \underline{D} X \underline{E} \rightarrow \underline{F}$  of morphisms a new composed morphism $g \circ_{X} f : \underline{D}\underline{A}\underline{E} \rightarrow \underline{B}\underline{F}\underline{C}$.
\item For each object $A$ and identity morphism $i_A : A \rightarrow A$.
\end{itemize}
Composition is subject to associativity and identity laws alongside:
\begin{itemize}
\item Interchange 1: \[ \tikzfigscale{0.7}{figs/correct_polyslotproof3} \quad = \quad \tikzfigscale{0.7}{figs/correct_polyslotproof4}. \]
\item Interchange 2:  \[ \tikzfigscale{0.7}{figs/correct_polyslotproof3flip} \quad = \quad \tikzfigscale{0.7}{figs/correct_polyslotproof4flip}. \]
\item Equivariance with respect to permutations: \[ \tikzfigscale{0.7}{figs/poly_equi_1} \quad = \quad \tikzfigscale{0.7}{figs/poly_equi_3}. \] More explicitly, $(\sigma g \rho) \circ_{X} (\lambda f \tau) = \lambda_{X \rightarrow \underline{A}} \sigma_{i_{B} (-)i_C} (g \circ_X f) \tau_{i_{D}(-)i_E} \rho_{X \rightarrow \sigma(\underline{A})} $ where for instance $\sigma_{i_{B}(-)i_{C}}$ means $i_{\underline{B} \otimes \sigma \otimes {i_{\underline{C}}}}$ and $\rho_{X \rightarrow \underline{A}}$ represents $\rho$ in which the role of $X$ is replaced by the entire list $\underline{A}$.

\end{itemize}

%Formally, whilst polycategories are not monoidal they still are equipped with a certain notion of interchange law. which is in this context often referred to as associativity. The full specification of the associativity conditions required for symmetric polycategories is left to the Appendix.

%Let us here be more explicit about symmetric polycategories, following \cite{shulman20202chudialectica}. A symmetric polycategory $\mathbf{P}$ is a collection of objects $A,B, \dots$ and for each pair of lists $\underline{A}$ and $\underline{B}$ a collection $\mathbf{P}(\underline{A}, \underline{B})$ of morphisms with:

\paragraph{Quantum Theory}
The category $\mathbf{fHilb}$ can be viewed as the fundamental raw-material category from which a multitude of categories relevant to quantum information processing can be constructed.
\begin{definition}[The category $\mathbf{fHilb}$]
The category $\mathbf{fHilb}$ has objects given by finite dimensional Hilbert spaces and morphisms given by linear maps. Sequential composition in $\mathbf{fHilb}$ is given by the standard composition rule for linear maps, the monodial product is given on objects by the tensor product $H_A \otimes H_B$ of Hilbert spaces. On morphisms the monodial product is given by linear extension of $(f \otimes g) ( \phi \otimes \psi) := f(\phi) \otimes g(\psi)$. The category $\mathbf{fHilb}$ is furthermore compact closed with $\cup := \sum_i \ket{i} \otimes \ket{i}$ and $\cap = \sum_{i}\bra{i} \otimes \bra{i}$. 
\end{definition}
The category $\mathbf{fHilb}$ can be viewed as the fundamental raw-material category from which a multitude of categories relevant to quantum information processing can be constructed. The main category we will be concerned with in this paper is the category that is typically interpreted as representing the time-reversible dynamics of quantum theory, the category $\mathbf{fU}$ of unitaries.
\begin{definition}[The category $\mathbf{fU}$]
The category $\mathbf{fU} \subseteq \mathbf{fHilb}$ has objects given by finite-dimensional Hilbert spaces and morphisms given by unitary linear maps, that is, linear maps $U:H_A \rightarrow H_B$ such that $U^{\dagger} \circ U = id = U \circ U^{\dagger}$. In this sense $U^{\dagger}$ is typically interpreted as the time-reverse of $U$. All sequential and parallel composition rules are inherited from $\mathbf{fHilb}$, however compact closure is not inherited since neither of $\cap,\cup$ or in general unitary.
\end{definition}
To account for noise, the category of (Unitary) linear maps is typically extended to the category of (Trace Preserving) completely positive maps. 
\begin{definition}[The category $\mathbf{fCP}$]
The category $\mathbf{fCP}$ has as objects the spaces $\mathcal{L}(H_A)$ of linear operators on Hilbert spaces. The morphisms of type $\mathcal{L}(H_A) \rightarrow \mathcal{L}(H_B)$ in  $\mathbf{fCP}$ are given by the completely-positive operators \cite{Wilde2011FromTheory}. $\mathbf{fCP}$ is also equipped with bell-states and effects and so is compact closed. The resulting isomorphism between states and processes in $\mathbf{fCP}$ is referred to as the CJ (Choi-Jamiolkowski) \cite{Choi1975CompletelyMatricesc} isomorphism.
\end{definition}
\begin{definition}[The category $\mathbf{fQC}$]
The category $\mathbf{fQC}$ of quantum channels is the sub-category of $\mathbf{fCP}$ containing only those morphisms which are trace-preserving. $\mathbf{fQC}$ is not compact closed since its only effect is the trace. The quantum channels are the processes in quantum information theory most commonly referred to as deterministic.
\end{definition}
In both the pure and mixed cases, we see that the deterministic evolutions arise by picking out a preferred subcategory of a compact closed category. This story carries over into the definition of higher-order deterministic evolutions called supermaps.

\paragraph{Quantum supermaps}
Quantum supermaps are used in quantum information theory and quantum foundations to formalize a notion of higher-order transformation that can be applied to transformations \cite{Chiribella2008TransformingSupermaps}. Intuitively the goal of the definition of quantum supermaps is to formalize the following kind of picture  \[  \tikzfig{figs/informal_poly3},   \] used to represent a higher-order process that accepts as an argument a process of type $A_1 \rightarrow A_1'$ and a process of type $A_2 \rightarrow A_2'$ to produce a process of type $B \rightarrow B'$. Such maps will typically be interpreted as having type $[A_1,A_1'][A_2,A_2'] \rightarrow [B,B']$ within some kind of algebraic structure. 

It is typically required that such maps should be well-defined when acting on parts of bipartite processes \[  \tikzfig{figs/informal_poly4} .  \] However, some care has to be taken in defining how supermaps can be used or composed together due to a key structural feature of supermaps termed \textit{enrichment} \cite{kelly1982basic} which is a mathematical translation of the idea that the basic structural features present in $\mathbf{C}$ (parallel and sequential composition) can be implemented as higher-order transformations \cite{Wilson2022AProcesses, Rennela2018ClassicalTheory}. In other words, we expect there to exist a supermap of type $\circ:[A,B][B,C] \rightarrow [A,C]$ which implements sequential composition viewed intuitively as \[  \tikzfig{figs/informal_poly5}  . \]

This simple observation, and a generally expected feature of theories of supermaps, motivates a further expected polycategorical feature of supermap composition. Polycategorical structure as witnessed by linear distributivity of the $\mathbf{Caus}[\mathbf{C}]$ construction allows us to unambiguously give meaning to
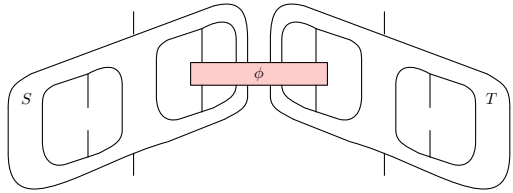
\[  \tikzfig{figs/informal_poly1}  \]
with the following diagram in the graphical language for polycategories
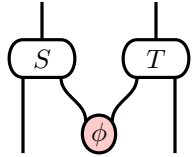
\[  \tikzfigscale{0.7}{figs/polysem1} . \]
Because polycategories can only be connected one leg at a time, there is no composition rule in the definition of a polycategory, which allows to give meaning to the following diagram 
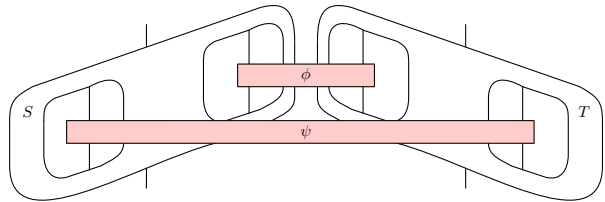
\[  \tikzfig{figs/informal_poly2},  \] which would require the possibility to compose along more than one wire at-once, creating cycles \[  \tikzfigscale{0.7}{figs/polysem2}.  \]
Such a cyclic diagram at the level of supermaps should not be allowed since when combined with the structure of enrichment, it could be used to produce time loops within the underlying category as intuitively represented by the following diagram \[  \tikzfig{figs/informalpoly_loop}.  \] These observations point us towards the following goal for constructions of  higher-order transformations in quantum theory, the construction of a polycategory which enriched the symmetric monoidal structure of the category it is intended to act upon. Constructions for supermaps on abstract symmetric monodial categories which fit within goal may then be composed in complex ways whilst guaranteeing that time-loops never be formed.

In the usual approach to the definition of quantum supermaps, the Choi-Jamilkiowski (CJ) isomorphism \cite{Choi1975CompletelyMatrices} is leveraged, which identifies completely positive maps with positive operators. Here we will review the standard definition of quantum supermaps in a way that allows to briefly point out their polycategorical structure. The definition we use slightly generalizes the construction of a polycategory of second-order causal processes using the $\mathbf{Caus}[\mathbf{C}]$ construction \cite{Kissinger2019AStructure} by never referencing the concept of causality and instead using a definition method provided in \cite{Wilson2022QuantumLocality}. This reference to causality prevents the $\mathbf{Caus}[\mathbf{C}]$ construction from giving a way to construct a unitary higher order quantum theory, however,  a sketch definition for the linearly distributive structure of unitary supermaps has been outlined in \cite{Wilson2022AProcesses}. In category-theoretic terms, the CJ isomorphism is the observation of compact closure of the category $\mathbf{CP}$ and it is compact closure when present which allows for a convenient definition of supermaps.
\begin{definition}[$\mathbf{P}$-Supermaps]
Let $\mathbf{C} \subseteq \mathbf{P}$ be an embedding of a symmetric monodial category $\mathbf{C}$ into a compact closed category $\mathbf{P}$, a morphism \[ \tikzfigscale{0.7}{figs/newsupermap_1}  \] in $\mathbf{P}$ is a $\mathbf{P}$-supermap on $\mathbf{C}$ of type $S:[A,A'] \rightarrow [B,B']$ if and only if for every $\phi \in \mathbf{C}(A \otimes X,A' \otimes X')$ then  \[ \tikzfigscale{0.7}{figs/newsupermap_ext} \quad \in \quad \mathbf{C}(A \otimes X,A' \otimes X'). \]
\end{definition}
When a category has states and effects a meaningful generalization can be given for supermaps of type $K \rightarrow M$ with $K \subseteq \mathbf{C}(A,A')$ and $M \subseteq \mathbf{C}(B,B')$ \cite{Wilson2022QuantumLocality}, however since there are no such states or effects in the category of unitaries we prefer to use the above definition which is less general but avoids pathological edge cases. The definition of supermaps can also be extended to the multi-party setting. 
\begin{definition}
Let $\mathbf{C} \subseteq \mathbf{P}$ be an embedding of a symmetric monodial category $\mathbf{C}$ into a compact closed category $\mathbf{P}$, a morphism \[ \tikzfigscale{0.7}{figs/newmultisupermap1}  \] in $\mathbf{P}$ is a $\mathbf{P}$-supermap on $\mathbf{C}$ of type $S:\Gamma \rightarrow [B,B']$ if and only if for every $\underline{A}_{i}:= [A_i,A_i']$ of $\Gamma$ and family ${\phi}_{i} \in \mathbf{C}(A_i \otimes X_i,A_i' \otimes X_i')$ then  \[ \tikzfigscale{0.7}{figs/newmultisupermap2} \quad \in \quad  \mathbf{C}(B \otimes X_1 \otimes \dots \otimes X_n ,B' \otimes X_1' \dots \otimes X_n') \]
\end{definition}

\begin{lemma}
A symmetric polycategory $\mathbf{polyPsup}[\mathbf{C}]$ can be defined with objects given by pairs $[A,A']$ of objects of $\mathbf{C}$ and morphisms of type $S: \Gamma \rightarrow \Delta$ given by the $\mathbf{P}$-supermaps of type $S:\Gamma \rightarrow \Delta$.
\end{lemma}
\begin{proof}
Given in the appendix.
\end{proof}
\begin{definition}[Quantum Supermaps]
For brevity we refer to the $\mathbf{fCP}$-supermaps on $\mathbf{fQC}$ as \textrm{Quantum Supermaps} and the corresponding polycategory is refered to as $\mathbf{QS}$, we furthermore refer to the $\mathbf{fHilb}$-supermaps on $\mathbf{fU}$ as \textrm{Unitary Supermaps} with the corresponding polycategory denoted $\mathbf{uQS}$.
\end{definition}

\paragraph{Locally-Applicable Transformations}
In this section we review a characterization of quantum supermaps as certain types of natural transformations \cite{Wilson2022QuantumLocality} called locally-applicable transformations (recently identified with \textit{strong} natural transformations \cite{hefford_prof_sup}), this removes the need to reference an ambient category such as $\mathbf{P}$ into which the category $\mathbf{C}$ embeds when defining supermaps. The goal of the paper will be to extend this natural transformation definition of supermap so that it is strong enough to (a) recover unitary supermaps when applied to the category of unitaries (b) extend supermaps to infinite dimensions in a satisfactory way (c) construct a (enrichment into a) polycategory for composition-without-time-loops of supermaps on arbitrary symmetric monoidal categories.

To introduce the concept, recall that higher-order transformations are modelled as those transformations that can be applied locally to lower-order transformations \[  \tikzfig{figs/informal_cp_3}.   \] Let us imagine then that from any legitimate definition of supermap we expect to be able to extract at a bare minimum a familly of functions $S_{X,X'}: \mathbf{C}(A \otimes X, A' \otimes X') \rightarrow \mathbf{C}(B \otimes X, B' \otimes X')$ where we will denote graphically the action of $S_{X,X'}$ on some $\phi \in \mathbf{C}(A \otimes X, A' \otimes X')$ as: \[ S_{X,X'}(\phi) \quad :=     \quad \tikzfig{figs/informal_cp_3b}.   \]
\begin{definition}[locally-applicable transformations]
A locally-applicable transformation of type $S:[A,A'] \longrightarrow [B,B']$ is a family of functions $S_{X,X'}$ satisfying \[ \tikzfig{figs/formal_cp_1aa} \quad = \quad \tikzfig{figs/fornal_cp_1bb}.  \]
\end{definition}
The locally applicable transformations define a category $\mathbf{lot}[\mathbf{C}]$ with objects given by pairs $[A,A']$ and morphisms $[A,A'] \rightarrow [B,B']$ given by locally applicable transformations of the same type. $\mathbf{P}$-supermaps on a category $\mathbf{C}$ always define locally-applicable transformations on $\mathbf{C}$, as witnessed by a faithful functor $\mathcal{F}_{\mathbf{P}}: \mathbf{Psup}[\mathbf{C}] \rightarrow    \mathbf{lot}[\mathbf{C}]$. This functor is given explicitly by \[ \mathcal{F}_{\mathbf{P}} \Bigg(  \tikzfigscale{0.7}{figs/final1}   \Bigg)_{XX'} \quad := \quad  \tikzfigscale{0.7}{figs/final9} \] In \cite{Wilson2022QuantumLocality} it is proven that there is an equivalence between the quantum supermaps and the locally-applicable transformations on $\mathbf{QC}$.
\begin{theorem}
There is a one-to-one correspondence between the locally-applicable transformations of type $[A,A'] \rightarrow [B,B']$ on $\mathbf{QC}$ and the quantum supermaps of the same type \cite{Wilson2022QuantumLocality}.
\end{theorem}

As we will observe in the main text of the paper, there is no such correspondence between the locally-applicable transformations on $\mathbf{U}$ and the Unitary supermaps. A stronger notion, that of being a \textit{slot} will be further required. We finish the preliminary material by noting that locally-applicable transformations admit a simple multi-party generalization.
 \begin{definition}
A locally-applicable transformation of type $[A_1,A_1'] \dots [A_n,A_n'] \rightarrow [B,B']$ is a family of functions $S_{E_1 \dots E_n}^{E_1' \dots E_n'}$ satisfying \[ \tikzfig{figs/multi_slot_1a} \quad = \quad \tikzfig{figs/multi_slot_1b}. \]
\end{definition}
These multiple-input locally-applicable transformations do not appear to come with a natural polycategorical structure, instead being equipped with a weaker notion of multi-categorical structure \cite{Leinster2003HigherCategories}. Multicategories allow for multiple inputs but only a single output, allowing to draw only the following kinds of string diagrams  \[  \tikzfigscale{0.7}{figs/multi_comp} \quad  \circ \quad \tikzfigscale{0.7}{figs/multi_comp2} \quad = \quad \tikzfigscale{0.7}{figs/multi_comp3}.  \].
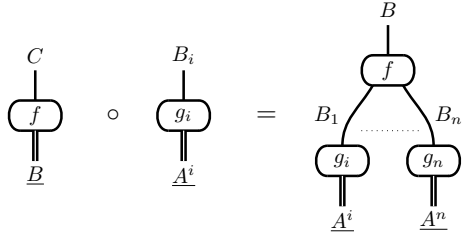
\begin{lemma}[The multicategory of locally-applicable transformations]
A multi-category $\mathbf{lot}[\mathbf{C}]$ can be defined which has as objects pairs $[A,A']$ and as multi-morphisms from $[A_1,A_1'] \dots [A_n,A_n'] $ to $[B,B']$ the locally-applicable transformations of type $[A_1,A_1']  \dots [A_n,A_n'] \longrightarrow [B,B']$. Composition is given graphically by taking $ S \circ (T^1 \dots T^m)(\phi_i^j)$ to be  \[  \tikzfig{figs/multi_slot_comp1} . \]
\end{lemma}
Since we have seen motivation for developing a polycategorical semantics for supermaps, the fact that it is only clear how to give locally-applicable transformations a multi-categorical structure is a sign that stronger conditions are required. This, is essentially the same issue as the inability to give a suitable monoidal product for locally-applicable transformations. The above difficulty is discussed in more detail in the next section, after which two strengthenings of locally applicable transformations are developed. Each of these strengthenings characterize the unitary supermaps and on arbitrary symmetric monoidal categories return polycategorical rather than multi-categorical composition rules.

\section{The Need for a Stronger Definition than Locally Applicable Transformation}
In this section we show why locally-applicable transformations are not strong enough to satisfy our two goals, in the course of doing so we introduce a few definitions which will be used throughout the paper.
We will introduce two conceptual and technical issues regarding locally-applicable transformations, and show that both problems can be addressed by strengthening them to \textit{strongly} locally-applicable transformations (slots). A slot will be morphism in the centre, suitably defined, of the locally-applicable transformations. 
%This definition solves by force the problem of the interchange law by definition and in turn by forcing commutation with transformations such as $S^{loop}$ and $S^{V}$ this definition will be strong enough to tightly characterize the unitary supermaps. In short we will find that for any $\mathbf{C}$ then $\mathbf{slot}[\mathbf{C}]$ is monoidal, $\mathbf{pslot}[\mathbf{C}]$ is a polycategory, and furthermore that $\mathbf{plsot}[\mathbf{fU}] = \mathbf{uQS}$ and $\mathbf{plsot}[\mathbf{fQC}] = \mathbf{QS}$. Meaning that the polycategorical structure of supermaps can be generalized to arbitrary symmetric monoidal categories.
\paragraph{Problem 1: Characterising Unitary Supermaps} Let us consider two locally-applicable transformations on the category of unitaries which will play an important role throughout this paper. Both classes are given by conditioning on properties of unitaries which due to the time-reversibility of $\mathbf{U}$ cannot be affected by applying local unitaries to auxiliary systems. The first example of a locally-applicable transformation works by checking the signalling structure of a unitary and applying a time-loop whenever the application of a time loop is permitted by the signalling structure of the unitary. 
\begin{definition}
The locally-applicable transformation $S^{loop}:[A,A'] \rightarrow [B,B']$ is defined by taking $S^{lopp}_{XX'}(\phi)$ to be
\begin{align} 
\tikzfig{figs/unitary_1} \quad  & := \quad   \tikzfig{figs/unitary_loop} \quad \textrm{if} \quad  \tikzfig{figs/unitary_same} \quad = \quad  \tikzfig{figs/phi_upsig3} \\
& := \quad   \tikzfig{figs/unitary_same} \quad \textrm{if else} .
\end{align}
\end{definition}
The second example uses the signalling structure of the input unitary to decide whether to apply a local unitary. 
\begin{definition}
The locally-applicable transformation $S^{V}:[A,A'] \rightarrow [B,B']$ is defined by taking $S^{V}_{XX'}(\phi)$ to be
\begin{align}
\tikzfig{figs/unitary_1} \quad  & := \quad   \tikzfig{figs/unitary_v} \quad \textrm{if} \quad  \tikzfig{figs/unitary_same} \quad =  \tikzfig{figs/phi_upsig3} \\
& := \quad   \tikzfig{figs/unitary_same} \quad \textrm{if else} 
\end{align}
\end{definition}

Each definition indeed gives a locally-applicable transformation on the category of unitaries. Neither of $S^{loop}$ or $S^{V}$ are however implementable by unitary supermaps.
\begin{lemma}
Let $S:[A,A] \rightarrow [A,A]$ be a $\mathbf{P}$-supermap on $\mathbf{U}$ such that $\mathcal{F}_{\mathbf{P}}(S) = S^{loop}$ or $\mathcal{F}_{\mathbf{P}}(S) = S^{V}$, then $A = I \cong \mathbb{C}$.
\end{lemma}
\begin{proof}
Assume that there exists some $S:A^{*} \otimes A' \rightarrow B^{*} \otimes B'$ such that \[      \tikzfig{figs/unitary_def_loop} \quad = \quad \mathcal{F}_{\mathbf{P}} \Bigg(  \tikzfigscale{0.6}{figs/final1}   \Bigg)_{XX'} \quad = \quad   \tikzfigscale{0.6}{figs/newsupermap_ext}. \]
For an arbitrary object $A$ consider the identity $id_A$, then \[      \tikzfig{figs/s_loop_ida} \quad = \quad   \tikzfig{figs/s_loop_id} \quad = \quad   \tikzfigscale{0.6}{figs/loop_standard1} \quad = \quad    \tikzfigscale{0.6}{figs/lopp_standard2}.  \]
Now, returning to function box representation
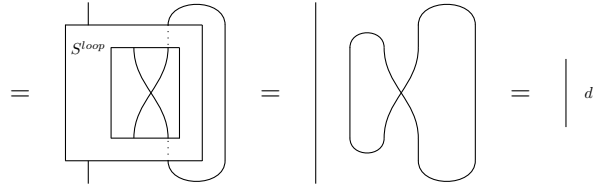
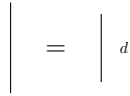
\[  \quad = \quad   \tikzfig{figs/s_loop_swap} \quad = \quad   \tikzfig{figs/output_swap}   \quad = \quad   \tikzfig{figs/id_d}   \]
It follows that \[ \tikzfigscale{0.8}{figs/s_loop_ida} \quad = \quad  \tikzfig{figs/id_d} ,  \] which in $\mathbf{U}$ is a contradiction unless $d_A = 1$ so that $A \cong \mathbb{C}$, A similar proof applies to the locally applicable transformation $S^{V}$. 
\end{proof}

\paragraph{Problem 2: Parallel Application of Supermaps}
Inuitively we imagine that given access to a bipartite process $\phi:A \otimes B \rightarrow A' \otimes B'$, one could imagine applying some supermap $S \boxtimes T$ which represents acting with $S:[A_1,A_1'] \rightarrow [A_2,A_2']$ on the left hand side and with $T:[B_1,B_1'] \rightarrow [B_2,B_2']$ on the right hand side \[   \tikzfig{figs/informal_par}. \] 
Now, let us imagine defining the application on the right hand side for $T$ by \[   \tikzfig{figs/badcomp1b} \quad = \quad \tikzfig{figs/badcomp1a}   \] One could hope to give meaning to the intuitive picture representing some notion of $(id \boxtimes S^{V}) \circ (S^{loop} \boxtimes id)$ by \[   \tikzfig{figs/informal2}  \quad \cong \quad \tikzfig{figs/badcomp4}. \] Analogously, we can write what we would hope to be the diagram representing $  (S^{loop} \boxtimes id) \circ (id \boxtimes S^{V})$ \[   \tikzfig{figs/informalint1}  \quad \cong \quad \tikzfig{figs/badcomp5} . \]
In a monoidal category these two terms would need to be the same, however, for the specific locally-applicable transformations chosen this is not the case. To seer this let us consider the action of each term on the swap. First note that \[   \tikzfig{figs/badcomp4swapa} \quad = \quad \tikzfig{figs/badcomp4swapb} \quad = \quad \tikzfig{figs/badcomp4swapc},  \]

and yet \[   \tikzfig{figs/badcomp5swapa} \quad = \quad \tikzfig{figs/badcomp5swapb} \quad = \quad \tikzfig{figs/badcomp5swapc}.  \]

Consequently, we observe the following, the locally-applicable transformations on unitaries which are \textit{not} unitary supermaps appear to be those which can be used to fail the interchange law. In the main contributions of this paper, we formalize this observation, showing that those locally-applicable transformations which are guaranteed to satisfy the interchange law, are exactly those which can be implemented as unitary supermaps. We call these (strongly) locally-applicable transformations \textit{slots}.

\section{Slots and Polyslots}
We motivate two constructions $\mathbf{slot}[\mathbf{C}]$ and $\mathbf{pslot}[\mathbf{C}]$ by the attempt to define the parallel composition of locally-applicable transformations. When trying to define such a parallel composition rule we will find that we need to still allow for auxiliary systems on a further third pair of wires, consequently we choose to introduce the following notation \[  \tikzfig{figs/braid_define1} \quad = \quad  \tikzfig{figs/braid_define2}.  \]
To construct from all locally-applicable transformations those which can be composed in parallel we consider only those $S$ which commute with all other locally applicable transformations $T$ in the following sense. 
%We will call these \textbf{S}trongly \textbf{LO}cally-applicable \textbf{T}ransformations \textit{slots}.
\begin{definition}
A slot of type $S:[A_1,A_1'] \rightarrow [A_2,A_2']$ is a locally-applicable transformation of the same type such that for every locally-applicable transformation $T:[B_1,B_1'] \rightarrow [B_2,B_2']$ and $\phi \in \mathbf{C}(A_1 \otimes B_1 \otimes X,A_1 ' \otimes B_1' \otimes X')$ then:
 \[ \tikzfig{figs/slot1} \quad = \quad \tikzfig{figs/slot2} \]
 The corresponding category $\mathbf{slot}[\mathbf{C}] \subseteq \mathbf{lot}[\mathbf{C}]$ is defined by keeping all objects and all locally-applicable transformations which are slots.
\end{definition}
So in intuitive terms, slots are those functions that are so local, that they commute not only with combs but with all other functions which commute with combs. Either of these commuting expressions can be used to define the parallel composition of slots. Intuitively, the monoidal product takes two slots $S,T$ and views them as a new single-slot $S \boxtimes T$ which can be used in the following way \[   \tikzfig{figs/informal_paraux}.   \] That both of the expressions in the definition of a slot are required to be equal guarantees unambiguous interpretation of the above picture and the required interchange law for symmetric monoidal categories.
\begin{theorem}
The category $\mathbf{Slot}[\mathbf{C}]$ is symmetric monoidal with: \begin{itemize}
    \item $[A,A']\boxtimes [B,B'] = [A \otimes B,A' \otimes B']$
    \item $(S \boxtimes T)_{X,X'}$ given by:  \[ \tikzfig{figs/slot1} \quad  \textrm{or equivalently} \quad \tikzfig{figs/slot2} \]
\end{itemize}
\end{theorem}
\begin{proof}
Given in the appendix. This is a special case of taking the centre of a premonoidal category \cite{premonoidal_power}, where in this case the premonoidal category at hand is $\mathbf{lot}[\mathbf{C}]$. 
\end{proof}
The definition of a slot can be generalized to slots with multiple inputs, which we pre-emptively refer to as \textit{polyslots}. From here on, when monoidal products of lists of wires or morphisms need to be expressed, we use doubled wires.
\begin{definition}[Multi-party slots]
Let $\underline{\mathbf{A}}$ be a list with each element of the form $\mathbf{A}_i = [A_i,A_i']$ for some objects $A_i,A_i'$ of $\mathbf{C}$, a polyslot of type $S:\underline{\mathbf{A}} \rightarrow [B,B']$ is a locally-applicable transformation of type $\underline{\mathbf{A}} \rightarrow [B,B']$ such that for every $k$ and every $\underline{\phi}_{1 \dots k-1},\underline{\phi}_{k+1 \dots |\underline{\mathbf{A}}|}$ then the family of functions given by \[ \tikzfig{figs/reduce_slot2} \quad := \quad \tikzfig{figs/reduce_slot} , \] is a slot of type \[ S^{i}(\phi_{(m)}):[A_i,A_i'] \rightarrow [B \otimes \underline{X}_{m<i} \otimes \underline{X}_{m>i}, B' \otimes \underline{X}_{m<i}' \otimes \underline{X}_{m>i}' ] . \] 
\end{definition}
\begin{theorem}
The polyslots on $\mathbf{C}$ define a polycategory $\mathbf{pslot}[\mathbf{C}]$ with:
\begin{itemize}
    \item Objects given by pairs $[A,B]$ with $A,B$ objects of $\mathbf{C}$
    \item Poly-morphisms of type $S:\underline{\mathbf{A}} \rightarrow \Theta$ given by polyslots of type $S:[A_1,A_1'] \dots [A_n,A_n'] \rightarrow [B_1 \otimes \dots \otimes B_m ,B_1' \otimes  \dots \otimes B_m']$
    \item Composition $T \circ_{M} S$ of $S:\underline{\mathbf{A}} \rightarrow \underline{\mathbf{B}} \mathbf{M} \underline{\mathbf{C}} $ and $S:\underline{\mathbf{D}} \mathbf{M} \underline{\mathbf{E}} \rightarrow \underline{\mathbf{F}}$ given by taking $T \circ_{M} S(d_{(i)},a_{(j)},e_{(k)})$ to be  \[ \tikzfig{figs/locrepoly2} \]
\end{itemize}
\end{theorem}
\begin{proof}
Given in the appendix.
\end{proof}

\subsection{Single-Party Representable Supermaps}
Here we give a minimal construction that generalizes the multiparty unitary and CPTP supermaps to arbitrary categories, the construction works by leveraging a structural theorem for unitary and CPTP supermaps, that they always decompose \textit{locally} as combs. We will find that this construction is a special case of the definition of polyslots.
\begin{definition}
A single-party representable supermap of type \[ S: [A_1,A_1'] \dots [A_N,A_N'] \rightarrow [B,B']  \] is a family of functions \[S_{X_1 \dots X_N,X_1 ' \dots X_N'}: \mathbf{C}(A_1 X_1,A_1 ' X_1 ') \dots \mathbf{C}(A_N X_N,A_N ' X_N ') \rightarrow \mathbf{C}(B X_1 \dots X_N, B_1 ' X_1'  \dots X_N')\] such that for every $i$ and family of morphisms $\phi_{(m)}$ with $m \in \{  1 \dots (i-1)(i+1) \dots n \}$ there exists $S(\phi_{(m)})_i^{u}$ and $S(\phi_{(m)})_i^{d}$ satisfying \[S_{X_1 \dots X_N,X_1 \dots X_N'}(\phi_1 \dots \phi_i \dots \phi_N) \quad  = \quad  \tikzfig{figs/locrep2} . \]
\end{definition}
\begin{lemma}
Single-party representable supermaps of type $ S: [A_1,A_1'] \dots [A_N,A_N'] \rightarrow [B,B']$ are locally applicable transformations of the same type.
\end{lemma}
\begin{proof}
We define \[ \tikzfig{figs/locreppsi1} \quad = \quad \tikzfig{figs/locreppsi2}, \]
and then use locally representability to say that
\[ \tikzfig{figs/locrep4} \quad = \quad \tikzfig{figs/locrep5}, \] where finally, using the interchange law for symmetric monoidal categories, we can write: \[  = \quad \tikzfig{figs/locrep6} \quad = \quad \tikzfig{figs/locrep7} . \] Going through the same steps for every $i$ completes the proof.
\end{proof}
We now note that single-party representable supermaps on $\mathbf{C}$ form a polycategory.
\begin{theorem}
The single-party representable supermaps on $\mathbf{C}$ define a polycategory $\mathbf{srep}[\mathbf{C}]$ with:
\begin{itemize}
    \item Objects given by pairs $[A,B]$ with $A,B$ objects of $\mathbf{C}$
    \item Poly-morphisms of type $S:\Gamma \rightarrow \Theta$ with $\Gamma = [A_1,A_1'] \dots [A_n,A_n']$ and $\Theta = [B_1,B_1'] \dots [B_n,B_n']$ given by single-party representable supermaps of type $S:[A_1,A_1'] \dots [A_n,A_n'] \rightarrow [B_1 \otimes \dots \otimes B_m ,B_1' \otimes  \dots \otimes B_m']$
    \item Composition defined in the same way as for $\mathbf{pslot}[\mathbf{C}]$
\end{itemize}
\end{theorem}
\begin{proof}
The composition rule is the same as that of $\mathbf{pslot}[\mathbf{C}]$ and so is associative/unital. What must be checked is that the composition is still single-party representable. A careful proof is omitted but is a direct consequence of the fact that combs are closed under composition \cite{Hefford2022CoendCombs}.
\end{proof}

\begin{lemma}
For any symmetric monoidal category $\mathbf{C}$ then $\mathbf{srep}[\mathbf{C}] \subseteq \mathbf{pslot}[\mathbf{C}]$, meaning that every single-party representable supermap of type $S:[A_1,A_1'] \dots [A_n,A_n'] \rightarrow [B_1 \otimes \dots \otimes B_m ,B_1' \otimes  \dots \otimes B_m']$ is a polyslot of the same type.
\end{lemma}
\begin{proof}
This follows from noting that each single-party representable supermap, when acting on its part of any of its input bipartite processes acts as a comb, which implies that it commutes with any other locally-applicable transformation.
\end{proof}
So, single-party representable supermaps, are a special case of the polyslots. We will find that when applied to unitaries of arbitrary dimension, however, the strong locality property of polyslots is strong enough to enforce single-party representability. To frame this result we will require a generalization of traced monoidal categories to path-contraction categories. 
%This will in turn give us a third way to define supermaps on the category of unitaries between seperable Hilbert spaces.

\section{Path Contraction Categories}
We now consider pathing constraints, using relations between a choice of input and output decomposition to specify the ways in which a morphism decomposes. A more detailed discussion is given in \cite{Wilson2022QuantumLocality}, however, for the purposes of this paper we will only need to address a primitive form of pathing constraint of interest in the foundations of quantum information processing \cite{Eggeling2002SemicausalSemilocalizable}. We say that
 \[ \tikzfig{figs/phismall} \quad \in \quad \mathcal{E}_{path}\Big(\tikzfig{figs/no_path_label_norm}\Big)  \] if and only if there exist processes $1,2$ such that  \[ \tikzfig{figs/phismall} \quad = \quad \tikzfig{figs/phi_upsig2} . \] In this sense, the processes $1$ and $2$ serve as a witness for the satisfaction of the pathing constraint by $\phi$. 
 %The intuition is that the relation $\tau$ serves to specify the systems between which it is permitted that there may be directed paths, indeed the decomposition above forbids the presence of a directed path from the bottom-right object to the top-left object. 
 Whilst the above form is the most common considered in quantum information processing, we will more often be concerned with pathing constraints of the following form \[ \tikzfig{figs/phismall} \quad \in \quad \mathcal{E}_{path}\Big(\tikzfig{figs/phi_upsig1}\Big)  \] which entails the following decomposition \[ \tikzfig{figs/phismall} \quad = \quad \tikzfig{figs/phi_upsig3}.  \]

A key step in our characterisation of slots on unitaries as unitary supermaps, will be to observe that all unitary slots preserve non-pathing constraints of the above form. To allow us to phrase our results in a general form we define a generalization of compact closed (or trace monoidal) categories which allow for contraction of input and output wires \textit{only} when the contraction is such that it returns a morphism in $\mathbf{C}$.
\begin{definition}
The no-pathing functor $np_{A \nrightarrow B}(-,=): \mathbf{C}^{op} \times \mathbf{C} \rightarrow \mathbf{Set}$ is defined by \[np_{A \nrightarrow B}(X,X') \quad :=  \quad \mathcal{E}_{path}\Bigg( \textrm{ }\tikzfig{figs/no_path_label} \textrm{ }\Bigg) . \]
\end{definition}
Path contraction categories are then taken to be those symmetric monoidal categories in which at-least the no-pathing morphisms can be contracted. 
\begin{definition}
A path-contraction category is a symmetric monoidal category $\mathbf{C}$ equipped with a functor $pc_{A}(-,=): \mathbf{C}^{op} \times \mathbf{C} \rightarrow \mathbf{Set}$ satisfying \[ np_{A \nrightarrow A}(X,X') \subseteq pc_{A}(X,X') \subseteq \mathbf{C}(AX,AX'),  \] and equipped with for each $A$ a natural transformation $\eta_{X,X'}:pc_{A}(X,X') \rightarrow \mathbf{C}(X,X')$ denoted in function-box notation as \[ \eta(\phi \in \mathcal{E}(\tau)) \quad := \quad  \tikzfig{figs/path_contraction1}    ,    \] satisfying \[ \tikzfig{figs/path_contraction2a} \quad =  \quad  \tikzfig{figs/path_contraction2b}        \quad \quad \textrm{and} \quad \quad  \tikzfig{figs/path_contraction3} \quad =  \quad  \tikzfig{figs/path_contraction4} .       \] \end{definition}
The above properties along with naturality are enough to ensure that contraction along any no-pathing process evaluates in an intuitive way, namely that \[ \tikzfig{figs/path_intuit1} \quad =  \quad \tikzfig{figs/path_intuit2}       \quad =  \quad  \tikzfig{figs/path_intuit3} \quad = \quad \tikzfig{figs/path_intuit4}    .    \]
Note that whenever a category can be equipped with a path-contraction structure for some functors $pc_{A}(X,X')$ then it can always be equipped with a path-contraction structure for the functors $np_{A \nrightarrow A}(X,X')$.

%In general we refer to either of the above equivalent diagrams by \[   \tikzfig{figs/path_contraction9} \] 

\begin{example}
Any compact closed category $\mathbf{P}$ is a path-contraction category with the required natural transformations given by using the cup and cap \[ \tikzfig{figs/path_contraction1}  \quad = \quad  \tikzfig{figs/path_contraction_cup}   , .    \] 
where we take \[  pc_{A}(X,X')   =   \mathbf{C}(AX,AX').   \]
Furthermore, for any symmetric monoidal subcategory $\mathbf{C} \subseteq \mathbf{P}$ with $\mathbf{P}$ compact closed we can instead inherit a path-contraction structure from the path-contraction $\eta$ of $\mathbf{P}$ by defining $pc_{A}(X,X') = \{ \phi \in \mathbf{C}(AX,AX') :  \ \eta_{X,X'}(\phi) \in  \mathbf{C}(X,X') \}$.
\end{example}
%As another direct corollary, any symmetric monoidal subcategory $\mathbf{C}$ of a traced monoidal category $\mathbf{P}$ is a path contraction category via its embedding into the free compact closed category over any traced monoidal category. 
Consequently, the category $\mathbf{fU}$ of finite dimensional unitaries is a path contraction category via its embedding into $\mathbf{fHilb}$, as is the category $\mathbf{fQC}$ of finite dimensional quantum channels via its embedding into $\mathbf{fCP}$.
Our motivation for working with path-contraction categories as opposed to for instance categories that embed into compact closed categories is the ease with which they allow us to simultaneously discuss categories that include infinite-dimensional quantum systems. We take $\mathbf{sepHilb}$ to be the category of bounded linear maps between separable Hilbert spaces, and furthermore take $\mathbf{sepU} \subseteq \mathbf{sepHilb}$ to be the subcategory of unitary linear maps.
\begin{lemma}
The category $\mathbf{sepU}$ of unitaries between seperable Hilbert spaces is a path-contraction category.
\end{lemma}
\begin{proof}
In $\mathbf{sepHilb}$ one can write the identity processes as the result of a limit called \textit{resolution of the identity} \[ \tikzfig{figs/limit1} \quad =  \quad \texttt{Lim}_{n \rightarrow \infty} \Sigma_{i = 1 }^{n} \textrm{ } \tikzfig{figs/limit2}. \]
Furthermore, $\mathbf{sepHilb}$ has the property that limits commute with sequential and parallel composition, this is sufficient for us to define path contraction by \[ \tikzfig{figs/path_intuit1} \quad =  \quad \tikzfig{figs/limit3}. \] This is well defined since 
\[ \tikzfig{figs/limit3} \quad = \quad \texttt{Lim}_{n \rightarrow \infty} \Sigma_{i = 1 }^{n} \textrm{ } \tikzfig{figs/limit8} \quad = \quad \texttt{Lim}_{n \rightarrow \infty} \Sigma_{i = 1 }^{n} \textrm{ } \tikzfig{figs/limit6}, \] and so when \[ \tikzfig{figs/limit4} \quad = \quad \tikzfig{figs/limit5}, \] we can say that \[  \tikzfig{figs/limit3} \quad = \quad \texttt{Lim}_{n \rightarrow \infty} \Sigma_{i = 1 }^{n} \textrm{ } \tikzfig{figs/limit7} \quad = \quad \texttt{Lim}_{n \rightarrow \infty} \Sigma_{i = 1 }^{n} \textrm{ } \tikzfig{figs/limit9} \quad = \quad  \tikzfig{figs/limit3b}. \]
An alternative way to observe path contraction for $\mathbf{sepHilb}$ is to note that the weak pseudo-functorial embedding $\texttt{trunc}[-]_{w}: \mathbf{sepHilb} \rightarrow \mathbf{Hilb}^{*}$ of $\mathbf{sepHilb}$ into the compact closed $2$-category $\mathbf{Hilb}^{*}$ is sufficiently well-behaved to define path-contraction by using cups and caps of $\mathbf{Hilb}^{*}$ \cite{Gogioso2019QuantumMechanics}. We instead give the construction in terms of limits explicitly since we expect such tools to be more familiar to the wider physics community.
\end{proof}
We do not ask that whenever $\phi \in pc_{Z}(YX,YX') \cap pc_{Y}(ZX,ZX')$ (up to swaps) and furthermore \[ \tikzfig{figs/path_contraction5} \quad \in \quad pc_{Y}(X,X') , \] then
\[ \tikzfig{figs/path_contraction6} \quad \in \quad pc_{Z}(X,X'),  \] and furthermore \[ \tikzfig{figs/path_contraction7} \quad = \quad \tikzfig{figs/path_contraction8}, \] where again swaps have been used to define the contraction of wires which are not on the left-hand side. This is because it is not obvious in the infinite-dimensional case whether the taking of such limits ought to commute. When a path-contraction category also satisfies this property we will refer to it as a commuting path contraction category, examples include those above which are constructed from symmetric monoidal subcategories $\mathbf{C} \subseteq \mathbf{P}$ of compact closed categories. 

Generally, path-contraction structure when present, can itself be used to construct a definition (or at least internal representation ansatz) for supermaps.
\begin{definition}
Let $\mathbf{C}$ be a path-contraction category with functor $pc_{A}(X,X')$, then a path contraction supermap of type $S:\Gamma \rightarrow [B,B']$ is any locally-applicable transformation of the same type which takes the form \[ S_{X_i,X_i {'}} \quad = \quad \tikzfigscale{0.7}{figs/pc_supermap_2} , \] for any order of application of contractions along the $X_i$\footnote{One could instead ask that there exists some order of contractions which implements the associated locally applicable transformation, we will find however that slots characterise on path contraction groupoids to representations in which the result is independent of the order in which contractions are taken.}.
\end{definition}
For any commuting path contraction category, the path-contraction supermaps define a polycategory $\mathbf{pathcon}[\mathbf{C}]$. This is a straightforward generalization of the proof for $\mathbf{P}$-supermaps with $\mathbf{P}$ a compact closed category. We note without proof from now on that given an embedding of a symmetric monoidal category $\mathbf{C}$ into a compact closed category $\mathbf{P}$ then the $\mathbf{P}$-supermaps are in one-to-one correspondence with the path-contraction-supermaps where the path-contraction functor is taken to be given by specifying (up to cups and caps) the set of all $\mathbf{P}$-supermaps.

Note that each of $\mathbf{fU}$ and $\mathbf{sepU}$ are groupoids. 
\begin{definition}
A path-contraction groupoid is a path-contraction category in which every morphism is an isomorphism.
\end{definition}
\begin{example}
$\mathbf{fU}$ and $\mathbf{sepU}$ are path-contraction groupoids.
\end{example}
Consequently the language of path-contraction groupoids will allow us to prove theorems simultaneously for unitaries on finite-dimensional, and seperable Hilbert spaces. We finish by noting the following property, which we already implicitly used to conlude that $S^{loop}$ and $S^{V}$ are well-formed locally-applicable transformations on the symmetric monoidal category of unitary linear maps.  
\begin{lemma}
In a groupoid, for every $V:Y \rightarrow X$ and $W:X' \rightarrow Y'$ then:
\[ \tikzfig{figs/unitary_same} \in \quad \mathcal{E}_{path}\Big(\tikzfig{figs/phi_upsig1flip}\Big) \iff  \tikzfig{figs/unitary_v_w} \in \quad \mathcal{E}_{path}\Big(\tikzfig{figs/phi_upsig1flip}\Big)  \]
\end{lemma}
\begin{proof}
Given by invertibility of $V,W$.
\end{proof}
This lemma allows to generalize the definitions of $S^V$ and $S^{loop}$ to arbitrary path-contraction groupoids.

\section{Characterisation of Polyslots on Path Contraction Groupoids}
Here we show that slots on path contraction groupoids can always be implemented by combs. We begin by showing that their action on swap morphisms always decomposes into a no-pathing morphism.

\begin{lemma}[Slots Preserve Signalling Constraints]
Let $S:[A_1,A_1'] \rightarrow [A_2,A_2']$ be a slot on a path-contraction groupoid $\mathbf{G}$ then \[  \tikzfig{figs/s_on_swap} \quad \in \quad  \mathcal{E}_{path}\Big(\tikzfig{figs/phi_upsig1_rev}\Big) . \]
\end{lemma}
\begin{proof}
Assume that \[  \tikzfig{figs/s_on_swap} \quad \not\in \quad  \mathcal{E}_{path}\Big(\tikzfig{figs/phi_upsig1_rev}\Big) ,  \] then using commutativity of $S$ with any $S^{V}$ with $V \neq id$ gives \[  \tikzfig{figs/lemma_sig_1} \quad = \quad  \tikzfig{figs/lemma_sig_2}  \quad = \quad  \tikzfig{figs/lemma_sig_3}  \quad = \quad  \tikzfig{figs/lemma_sig_4} .  \] Using the fact that every morphism in $\mathbf{G}$ is an isomorphism we then find that
\[ \implies \tikzfig{figs/lemma_sig_5} \quad = \quad  \tikzfig{figs/lemma_sig_6},   \]
and furthermore any path-contraction groupoid $\mathbf{G}$ we have $i \otimes U = i \otimes W \implies U = W$.
\end{proof}
Note that $S^{loop}$ cannot be a slot, since it fails to satisfy the above condition, of preserving non-pathing constraints. Whilst the swap satisfies a non-pathing constraint \[ \tikzfig{figs/swap_only} \quad  \in \quad \mathcal{E}_{path}\Big(\tikzfig{figs/phi_upsig1_rev}\Big) ,   \]
The action of $S^{loop}$ on the swap gives a signalling channel
\[ \tikzfig{figs/s_loop_swap_2} \quad = \quad  \tikzfig{figs/s_loop_ida} \quad \tikzfig{figs/s_loop_ida} \quad  \not\in \quad   \mathcal{E}_{path}\Big(\tikzfig{figs/phi_upsig1_rev}\Big) . \]
We now give our main theorem, that slots on path-contraction groupoids are always combs, meaning that polyslots are always single-party representable.
\begin{theorem}For any path-contraction groupoid $\mathbf{G}$ then $\mathbf{pslot}[\mathbf{G}] = \mathbf{srep}[\mathbf{G}]$.
\end{theorem}
\begin{proof}
We use the fact that the action of $S$ on the swap must be non-pathing. Let $U_1,U_2$ be morphisms which witness this non-pathing constraint, then using the fact that $\mathbf{G}$ is a path-contraction category we can say that \[\tikzfig{figs/comb_1} \quad = \quad   \tikzfig{figs/comb_2path}  \quad = \quad   \tikzfig{figs/comb_3path} . \] Now using the diagrammatic rules for locally-applicable transformations this in turn in equal to
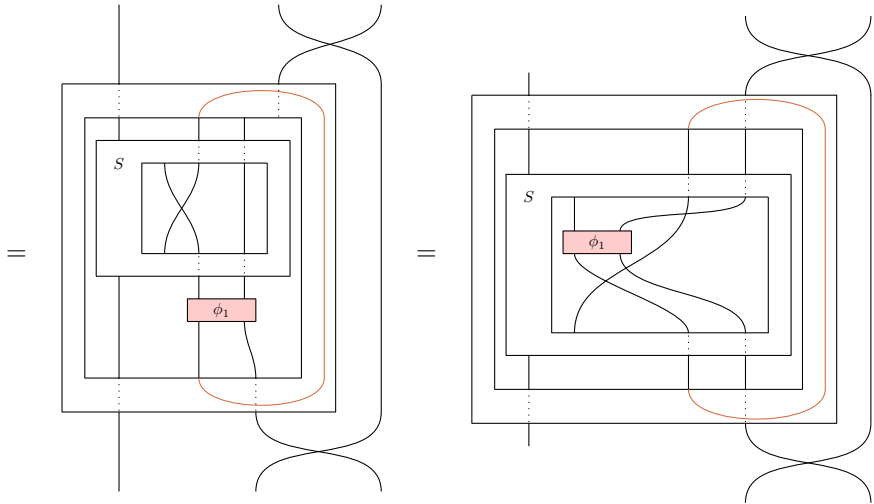
\[    \quad \tikzfig{figs/comb_4path}  \quad = \quad \tikzfig{figs/comb_5path}.   \] Then, using the definition of $S^{loop}$ and the fact that $S$ is a slot, the above in turn is equal to
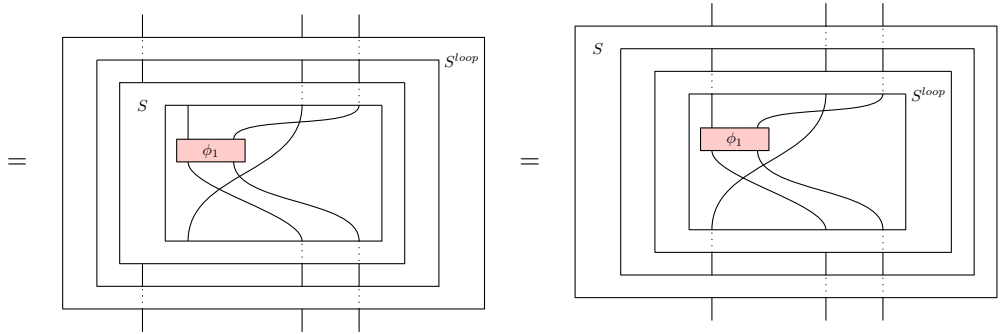
\[  \quad \tikzfig{figs/comb_6path}  \quad = \quad \tikzfig{figs/comb_7path}  .  \]
Then unpacking the definition of $S^{loop}$ and using the laws for path-contraction categories and locally-applicable transformations gives
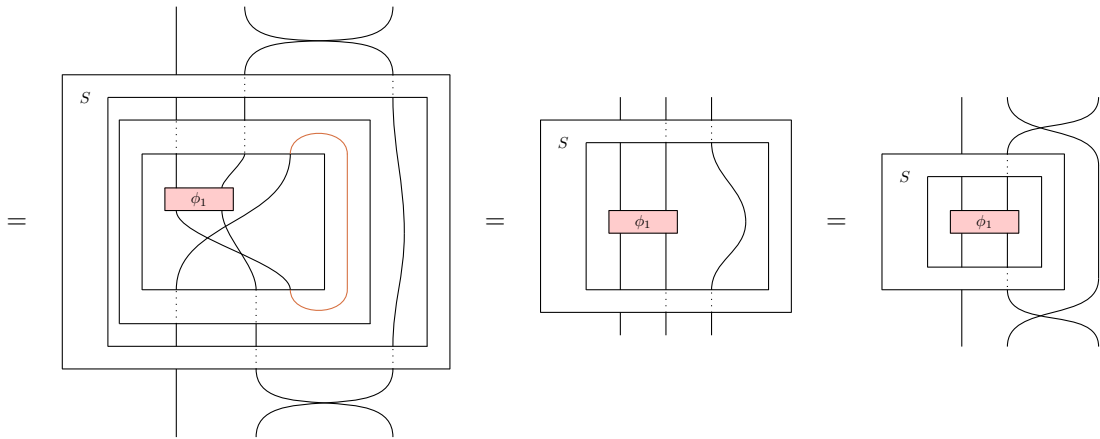
\[ = \quad   \tikzfig{figs/comb_8path} \quad = \quad \tikzfig{figs/comb_9path}   \quad = \quad \tikzfig{figs/comb_10path}. \] Finally, since $\mathbf{G}$ is a path contraction category this entails that \[\tikzfig{figs/comb_1a} \quad = \quad   \tikzfig{figs/comb_11} .  \] So far, we have proven that any slot is given by a comb, now we consider the case of a general multi-input polyslot. Focusing on some $\phi_i$, we examine the family of functions $S^{i}(\phi_i) := S(\phi_1 \dots \phi_i \dots \phi_n)$ where since $S$ is a polyslot each $S^i$ is by definition a slot and so by the above must decompose as a comb. Since this is true for each $i$, the slot $S$ is in-fact single-party representable.
\end{proof}

\begin{theorem}
For any path contraction category $\mathbf{C}$, every single-party representable supermap can be represented by a path-contraction supermap.
Concretely, any single-party representable supermap $S:[A_1,A_1'] \dots [A_n,A_n'] \rightarrow [B,B']$ on a path contraction category $\mathbf{C}$ can be implemented in terms of a process $S^{int}:A_1' \dots A_n' B \rightarrow A_1 \dots A_n B'$ of $\mathbf{C}$ and path-contractions in the following way: \[ \tikzfig{figs/multiloc} \quad = \quad \tikzfig{figs/internalproof0}  .\]
\end{theorem}
\begin{proof}
We give the proof for $N=2$, the extension to general $N$ is conceptually identical only heavier in notation. Define the required internal process by \[ \tikzfig{figs/internalproof0b} \quad := \quad  \tikzfig{figs/internalproof1} . \] Now, we evaluate the expression \[  \tikzfig{figs/internalproof2} . \] Without loss of generality let us imagine that that contraction along party $1$ is taken first and for simplicity study the $2$-input case, by the single-party representability property we can see that the above is equal to \[  \tikzfig{figs/internalproof3},  \] and using the fact that wlg we took the contraction along party $1$ first this is in turn equal to \[  \tikzfig{figs/internalproof4},  \] and hence \[  \tikzfig{figs/internalproof5}.  \] Finally, undoing local-representability gives \[ \tikzfig{figs/internalproof6},  \] and using analogous steps for $\phi_2$ gives the result.
\end{proof}

In general then, observing that consequently in any commuting path-contraction category $\mathbf{srep}[\mathbf{C}] \subseteq \mathbf{pathcon}[\mathbf{C}]$ we have that for any commuting path-contraction groupoid $\mathbf{pslot}[\mathbf{G}] = \mathbf{srep}[\mathbf{G}] \subseteq \mathbf{pathcon}[\mathbf{G}]$. As we will now see, each of these constructions generalizes the finite-dimensional unitary supermaps. 
%Meaning that we have a generalization of unitary supermaps to arbitrary dimensions $\mathbf{pslot}[\mathbf{C}]$ which (i) does not assume decomposition into combs at the single-party level (ii) only requires knowledge of the symmetric monoidal structure of $\mathbf{C}$ to be specified (iii) when applied to finite-dimensional unitaries recovers the standard definition unitary-preserving supermaps. In combination with previous results for quantum channels, this observation constitutes our main result.
\begin{theorem}Polyslots generalize quantum supermaps on the quantum channels and on the unitaries to arbitrary symmetric monoidal categories. Formally, there is an equivalence \[ \mathbf{pslot}[\mathbf{fU}] \cong \mathbf{uQS} \] of polycategories for the unitary case and an equivalence \[ \mathbf{pslot}[\mathbf{fQC}] \cong \mathbf{QS} \] of polycategories for the mixed case.
\end{theorem}
\begin{proof}
Based on the equivalence between path-contraction supermaps and $\mathbf{P}$-supermaps with $\mathbf{P}$ compact closed we have that $\mathbf{uQS} \cong \mathbf{pathcon}[\mathbf{fu}]$ and so $\mathbf{pslot}[\mathbf{fU}] = \mathbf{srep}[\mathbf{fU}] \subseteq \mathbf{pathcon}[\mathbf{fU}] \cong \mathbf{uQS}$. What remains is to show that $\mathbf{uQS} \subseteq \mathbf{srep}[\mathbf{fU}]$. In short, we must show that every unitary-preserving quantum supermap decomposes at the single-party level as a comb. First, every quantum supermap of type $[A,A'] \rightarrow [B,B']$ decomposes as a comb, which in graphical terms means that any $\mathbf{CP}$-supermap on $\mathbf{QC}$ decomposes as \[ \tikzfigscale{0.6}{figs/newsupermap_ext} \quad  =  \quad \tikzfig{figs/locrepcomb},  \] where $S^u$ and $S^d$ are quantum channels $\in \mathbf{fQC}$. A proof of this fact can be found in  \cite{Chiribella2008TransformingSupermaps} which at its core relies on the causal decomposition theorem for no-signalling channels \cite{Eggeling2002SemicausalSemilocalizable}. The fact that furthermore every single-party unitary supermap decomposes as a unitary comb is given in \cite{Yokojima2021ConsequencesSuperchannels}. In graphical terms this means that any $\mathbf{fHilb}$-supermap on $\mathbf{fU}$ decomposes as \[ \tikzfigscale{0.6}{figs/newsupermap_ext} \quad  =  \quad \tikzfig{figs/locrepcomb}  \] where $S^u$ and $S^d$ are unitaries $\in \mathbf{fU}$. As a consequence of the former decomposition theorem every quantum supermap of type $[A_1,A_1'] \dots [A_n,A_n'] \rightarrow [B,B']$ satisfies: \[ \tikzfigscale{0.6}{figs/locrep_sup} \quad  = \quad  \tikzfig{figs/locrep2}, \]
where the $S(\phi_{(m)})_i^u$ and $S(\phi_{(m)})_i^u$ are quantum channels. Furthermore, the same may be said for unitary supermaps, which can be shown to be realised in the same way by unitary linear maps. This can be shown by noting that fixing all but $\phi_i$, the resulting map $S(\phi_1, \dots \phi_{i-1} (-) \phi_{i+1} \dots \phi_N)$ defines up to braiding a single party supermap, so by the previous lemma must decompose as a comb. Consequently, we see that from any multiparty unitary supermap we can construct a single-party representable locally-applicable transformation.

Finally, the equivalence $\mathbf{pslot}[\mathbf{fQC}] \cong \mathbf{QS}$ follows from noting that since the locally-applicable transformations of type $\hat{S}:[A,A'] \rightarrow [B,B']$ are always given by $\hat{S} = \mathcal{F}_{\mathbf{QC}}(S)$ for some quantum supermap of the same type \cite{Wilson2022QuantumLocality}, then the slot condition for $\hat{S}$ (commutation) is inherited by the interchange law of $\mathbf{fCP}$. 
\end{proof}
Regarding the case of infinite dimensional supermaps, since $\mathbf{sepU}$ is a path contraction groupoid we already know that $\mathbf{pslot}[\mathbf{sepU}] = \mathbf{srep}[\mathbf{sepU}]$ and that every single-party representable supermap is a path contraction supermap. What is not so clear is whether there exists an infinite-dimensional analog of the canonical decomposition theorem for supermaps, that all possible path-contraction supermaps decompose at the single-party level as combs.

\section{Application: Quantum Switch for Hilbert Spaces of Arbitrary Dimension}
On the category $\mathbf{U}$ of unitaries between arbitrary Hilbert spaces, even beyond those which are separable, we can show that $\mathbf{pslot}[\mathbf{U}]$ and $\mathbf{srep}[\mathbf{U}]$ are broad enough to include generalisations of the quantum switch. We call a set $\{ \pi_k \} \subseteq \mathbf{Hilb}(Q,Q)$ a \textit{control} if $\pi_k \circ \pi_l = \delta_{k,l}$.
\begin{definition}[The Quantum Switch for Arbitrary Hilbert Spaces]
The quantum switch on $\mathbf{U}$ with control $\{ \pi_0,\pi_1 \}$ is defined as a polyslot of type $\texttt{Switch}:[A,A][A,A] \rightarrow [Q \otimes A,Q \otimes A]$ given by: \[   \tikzfig{figs/exampleswitch1} \quad = \quad\tikzfig{figs/exampleswitch2} \quad + \quad \tikzfig{figs/exampleswitch3} \] Where $\pi_0 = \ket{0} \bra{0}$ and $\pi_1 = \ket{1} \bra{1}$.
\end{definition}
$\texttt{Switch}$ is a single-party representable polyslot since its action on $\phi_2$ can be written as: \[ \tikzfig{figs/exampleswitchcomb} \]

Where \[   \tikzfig{figs/exampleswitch4pre} \quad = \quad\tikzfig{figs/exampleswitch4} \quad + \quad \tikzfig{figs/exampleswitch5} \] and \[   \tikzfig{figs/exampleswitch4post} \quad = \quad\tikzfig{figs/exampleswitch6} \quad + \quad \tikzfig{figs/exampleswitch7} \] and similarly for the action on $\phi_1$. This definition naturally extends to N-party switches of type $[A,A] \dots [A,A] \rightarrow [Q \otimes A, Q \otimes A]$, it is the conjecture of the authors that all unitary preserving supermaps including those with break causal inequalities admit indefinite dimensional analogues which are polyslots and so single-party representable.

\section{Summary}
The construction $\mathbf{pslot}[\mathbf{C}]$ satisfies a series of conditions which makes it a suitable generalization of the construction of quantum supermaps to arbitrary symmetric monoidal categories.
\begin{itemize}
    \item The definition of $\mathbf{pslot}[\mathbf{C}]$ only references the symmetric monoidal structure of $\mathbf{C}$,
    \item The definition of $\mathbf{pslot}[\mathbf{C}]$ does not assume the decomposition of supermaps into combs when viewed by individual parties, instead, this property is derived by the principle of locality,
    \item $\mathbf{pslot}[\mathbf{C}]$ is a symmetric polycategory into which $\mathbf{C}$ is enriched, which allows for sequential and parallel composition without allowing the formation of time-loops.
    \item $\mathbf{pslot}[\mathbf{C}]$ generalises the construction of unitary and standard quantum supermaps to arbitrary symmetric monoidal categories in the sense that $\mathbf{pslot}[\mathbf{fU}] = \mathbf{uQS}$ and $\mathbf{pslot}[\mathbf{fQC}] = \mathbf{QS}$.
\end{itemize}
Consequently, polyslots have a variety of properties making them suitable for the analysis and definition of supermaps for infinite-dimensional systems. A series of structural theorems guarantee the local realisability of polyslots as combs and the global realisability of polyslots by general internal processes with path-contraction, along with their inherited linearity. Left open is the question of whether in the case of $\mathbf{C} = \mathbf{sepU}$ the polyslots include all possible supermaps that could be defined by applying time-loops to unitaries on Hilbert spaces of separable dimension. Finally, polyslots are broad enough to include infinite-dimensional generalizations of canonical processes of interest such as the quantum switch, consequently polyslots provide a theory-independent definition of supermap with nice enough properties in the quantum realm to provide a potentially handy toolbox in the extension of the study of indefinite causal structure to infinite dimensions.
There are a variety of natural ways in which the work of this paper could be built upon
\begin{itemize}
    \item Whilst the language used in this paper is that of category theory, the theorems proven use the technology of string diagram rewriting. It is an open question as to whether the results of this paper can be viewed as consequences of more higher-level categorical arguments. A partial route to an answer might be the identification of supermaps which locally-decompose as combs and their polycategorical structure as arising from the structure of the preduals in the strong Hyland envelope of the Yoenda embedding of Coend Optics into the category of strong profunctors \cite{hefford2025bvcategoryspacetimeinterventions}. It is an open question as to whether the black-box definition of polyslots can arise in a similar way, and whether the equivalences between black-box and concrete definitions on path-contraction groupoids can also arise from more abstract reasoning regarding the categorical properties of the strong Hyland envelope. 
    \item There are important compositional features of supermaps beyond those inherited by polycategorical semantics, as discovered in \cite{ApadulaNo-signallingStructure}. It is again an open question as to whether such rich compositional semantics is available to the abstract constructions developed here, or whether instead, those compositional features are specific to the structure of quantum theory.
    \item Now that we have a well-behaved definition of supermaps for arbitrary OPTs including infinite-dimensional quantum theory, there is the question of whether the multitude of information processing advantages of supermaps with indefinite causal structure \cite{Ebler2018EnhancedOrder, Araujo2014QuantumOperations, Chiribella2009QuantumStructure, Chiribella2012PerfectStructures, Salek2018QuantumOrders, Chiribella2018IndefiniteChannel, Wilson2020ASwitches, Chiribella2020QuantumOrders, Sazim2020ClassicalChannels} extend past the finite-dimensional quantum-theoretic setting. This question will allow us to develop our understanding of the information processing advantages afforded by theories of quantum gravity.
    \item Further to the above point, it will be important to discover whether the construction of unitary-preserving supermaps from routed graphs \cite{Vanrietvelde2022ConsistentOrder} extends to the construction of polyslots in $\mathbf{sepU}$, so that canonical processes studied in quantum foundations can be lifted to the infinite-dimensional setting. This will require a generalization of polyslots to those which act on compositionally constrained spaces \cite{VanrietveldeRoutedCircuits, Wilson2021ComposableConstraints}
    \item It is unclear in the infinite-dimensional case whether one can find further physically reasonable supermaps by the generalization of the definition of supermaps in compact closed category to a definition of path-contraction supermaps. A proof of the conjecture that path-contraction supermaps in unitary quantum theory are equivalent to polyslots would suggest that a stable, circuit theoretic definition of supermap has been found. 
    %A proof of the converse would suggest that more work could be done to find a less strong circuit-theoretic definition of locality. In the mixed setting, there is even more, to be understood. All that can be known from the results of this paper for the case of the category $\mathbf{sepQC}$ of quantum channels on separable Hilbert spaces is that $\mathbf{srep}[\mathbf{sepQC}] \subseteq \mathbf{pslot}[\mathbf{sepQC}]$ and that $\mathbf{srep}[\mathbf{sepQC}] \subseteq \mathbf{pathcon}[\mathbf{sepQC}]$. Namely, the precise relationship between $\mathbf{pslot}[\mathbf{sepQC}]$ and $\mathbf{pathcon}[\mathbf{sepQC}]$ is unknown, if they differ then there may not be one most-appropriate definition of supermap for separable dimensions.
    \item Another open question is whether the relationship between the causal box framework \cite{Portmann2015CausalComposition} and the process matrix framework used to establish the possibility of embedding of processes with indefinite causal structure into a definitely ordered spacetime \cite{Vilasini2022EmbeddingMatrices}, extends to infinite-dimensional polyslots. The causal box framework, being phrased in terms of Fock space is indeed already expressed in a form suitable for the consideration of infinite dimensions.  
%    \item It is an open question as to whether polyslots as defined here either appear in, or are of use to, the other scientific disciplines in which black-box holes are studied.
    \item Whilst polyslots freely reconstruct supermaps, they cannot be used in the current form to freely construct all iterated layers of higher order quantum theory \cite{Bisio2019TheoreticalTheory,Kissinger2019AStructure}. A generalization of polyslots to those which in-fact act on polycategories appears to be required for such an iteration. 
\end{itemize}
More broadly, a circuit-theoretic black-box approach to holes in diagrams along with appropriate compositional rules has been proposed. Concrete holes (combs) appear outside of physics to as outlined in the introduction, leading naturally to the question of whether these less concrete black box-holes might also find application outside of the foundations of physics. 

\subsubsection*{Acknowledgements}
MW is grateful to A Vanreitvelde for suggesting use of single-party representability as an axiom from which supermaps could be reconstructed, J Hefford for noting that slots are the centre of the premonoidal category of locally-applicable transformations, and to A Kissinger for useful conversations regarding the linearity of locally-applicable transformations. GC is supported by the Chinese Ministry of Science and Technology (MOST) through grant 2023ZD0300600.. The opinions expressed in this publication are those of the authors and do not necessarily reflect the views of the John Templeton Foundation. GC was supported by the Croucher Foundation and by the Hong Kong Research Grant Council (RGC) though the Senior Research Fellowship Scheme SRFS2021-7S02. MW was supported by University College London and the EPSRC Doctoral Training Centre for Delivering Quantum Technologies. 

%\subsubsection*{Note Added}
%The observation of the polycategorical composition rule for polyslots and proofs of associativity
%were made independently for the specialized case of N-combs in \cite{Hefford2022CoendCombs}. 

\bibliographystyle{utphys.bst}
\bibliography{ref_local_q}
\appendix

\section{Polycategory of $\mathbf{P}$-supermaps}
We will find that when dealing with listed data naive diagrammatic representations become cumbersome, so for readability, we adopt a convention analogous to the convention used for genuine lists in multi/polycategories, choosing for instance to represent the following diagram  \[ \tikzfigscale{0.7}{figs/newmultisupermap2}, \] by: \[ \tikzfigscale{0.7}{figs/newsupermaplist}.  \] Such a language is not formalized but is used to convey the essence of proofs, with the unpacking of details left to the interested reader with access to larger pieces of paper.
\begin{lemma}
A symmetric polycategory $\mathbf{polysup}[\mathbf{P},\mathbf{C}]$ can be defined with objects given by pairs $[A,A']$ of objects of $\mathbf{C}$ and morphisms of type $S: \Gamma \rightarrow \Delta$ given by the $\mathbf{P}$-supermaps of type $S:\Gamma \rightarrow \Delta$, the composition rule is given by taking: \[ \tikzfigscale{0.7}{figs/polypsup1} \quad \circ_{\mathbf{M}} \quad \tikzfigscale{0.7}{figs/polypsup2}  \] to be \[ \tikzfigscale{0.7}{figs/polypsup3}  \] and with symmetric action by permutations given by: \[ \tikzfigscale{0.7}{figs/polysym2} \quad = \quad  \tikzfigscale{0.7}{figs/polysym3} \] \end{lemma}
\begin{proof}
This composition rule returns a new $\mathbf{P}$-supermap since the application of $T \circ_M S$ can be written \[ \tikzfigscale{0.7}{figs/polysym_welldf}\] which by the interchange law for symmetric monoidal categories can be converted to
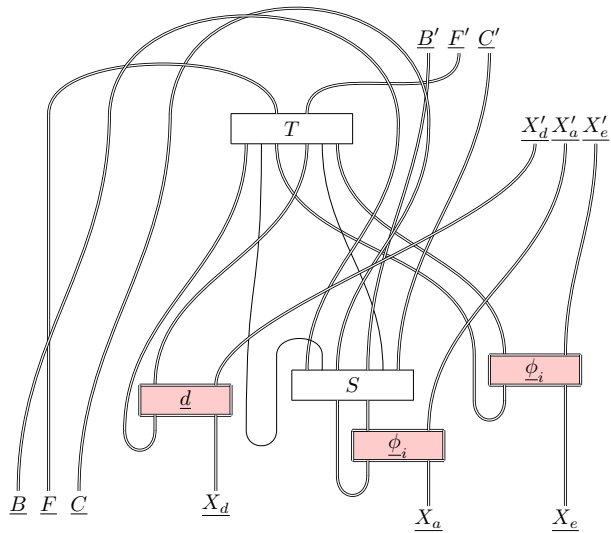
\[ \tikzfigscale{0.7}{figs/polysym_welldf2}\]
where since $S$ is a $\mathbf{P}$-supermap we can replace the action of $S$ by a new morphism $S'(a)$ of $\mathbf{C}$ to give
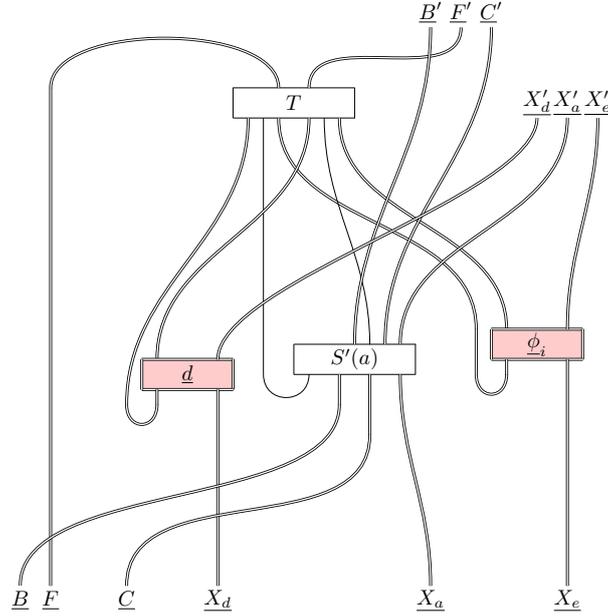
\[ \tikzfigscale{0.7}{figs/polysym_welldf3}\] what remains is the actions of $T$ on a series of channels with $B,C$ considered as extensions of the morphism $S'(a)$, consequently, the entire global diagram is a morphism of $\mathbf{C}$. 
The requires interchange laws for symmetric polycategories are satisfied as they are inherited directly from the interchange laws and symmetry of the symmetric monoidal structure of $\mathbf{P}$.
\end{proof}
It is noted in the main text that composition along multiple wires ought not to be allowed, so as to avoid the creation of time-loops, this point can be made at a more technical level now an explicit definition of supermap has been given.
A simple example demonstrates why two-wire composition rules are in general forbidden. Since $\mathbf{C}$ is a symmetric monoidal category, for any $\mathbf{C} \subseteq \mathbf{P}$ with $\mathbf{P}$ compact closed then there exists a $\mathbf{P}$-supermap of type $S:[A,A][A,A] \rightarrow [A,A]$ which performs sequential composition:  \[ \tikzfigscale{0.7}{figs/polywhy1form} \quad = \quad  \tikzfigscale{0.7}{figs/polywhy2form} \]
This is indeed a supermap since for all $\phi_1,\phi_2$ then: \[ \tikzfigscale{0.7}{figs/polywhy1b} \quad = \quad  \tikzfigscale{0.7}{figs/polywhy3new} \quad = \quad  \tikzfigscale{0.7}{figs/polywhy4}  \] which since $\mathbf{C}$ is a symmetric monoidal category must be in $\mathbf{C}$. Next note that there exists a $\mathbf{P}$-supermap of type $\phi:\emptyset \rightarrow [A,A][A,A]$ given by: \[ \tikzfigscale{0.7}{figs/polywhystateform} \] Indeed note that it is a supermap since the following \[  \tikzfigscale{0.7}{figs/polywhystate2} \quad = \quad  \tikzfigscale{0.7}{figs/polywhystate3}  \] is a member of $\mathbf{C}$ given that $\mathbf{C}$ is symmetric monoidal. However, if we were to try to compose $\phi$ and $S$ along both of their output/input wires, to give meaning to the following diagram \[ \tikzfigscale{0.7}{figs/polyfailform}  \] then a loop would be formed:  \[ \tikzfigscale{0.7}{figs/polywhy5} \quad = \quad  \tikzfigscale{0.7}{figs/polywhy6} \] There is no guarantee that this re-normalisation by a scalar preserves membership of $\mathbf{C}$, indeed in the study of quantum causal structure such loops are often interpreted as time-loops, and in the category $\mathbf{U}$ we find that such a re-normalisation does not preserve membership of $\mathbf{U}$. In the above sense we can see that the natural emergence of a polycategorical semantics can be understood as a compositional semantics which prevents the forming of time-loops.
%%%%%%%%%%%%%%%%%%%%%%

%%%%%%%%%%%%%%%%%%%%%%%%%%%

\section{Monoidal category of Slots}
To express the slot condition algebraically and proove symmetric monodial structure, we will find it easier to talk about for each $T$ the induced transformation $ (\beta T \beta )_{A_1 X}^{A_1^{'} X'} :=  \beta_{B_2 A_1}^{B_2^{'}A_1^{'}} T_{A_1 \otimes X}^{A_1^{'} \otimes X'} \beta_{ A_1 B_1}^{A_1^{'} B_1^{'}}$ defined by taking $\beta_{ABX'}^{A'B'X} := \mathbf{C}(\beta_{A B} \otimes X, \beta_{A^{'} B^{'}} \otimes X)$ and so then: \[  \begin{tikzcd}
{\mathbf{C}(A_1 \otimes B_1 \otimes X,A_1' \otimes B_1' \otimes X')} \arrow[rr, "{\mathbf{C}(\beta_{A_1 B_1} \otimes X, \beta_{A_1^{'} B_{1}^{'}} \otimes X)}", bend left] \arrow[d, "(\beta T \beta )_{A_1 X}^{A_1^{'}  X'}"', bend right] &  & {\mathbf{C}(B_1 \otimes A_1 \otimes X,B_1' \otimes A_1' \otimes X')} \arrow[d, "T_{A_1 \otimes X}^{A_1^{'} \otimes X'}", bend left]                                        \\
{\mathbf{C}(A_1 \otimes B_2 \otimes X,A_1' \otimes B_2' \otimes X')}                                                                                                                                                                        &  & {\mathbf{C}(B_2 \otimes A_1 \otimes X,B_2' \otimes A_1' \otimes X')} \arrow[ll, "{\mathbf{C}(\beta_{B_2 A_1} \otimes X, \beta_{B_2^{'} A_{1}^{'}} \otimes X)}", bend left]
\end{tikzcd}  \]

Note that for now we assume our underlying monoidal category is strict so that we do not have to keep track of associators and unitors.
\begin{theorem}
A monoidal category $\mathbf{slot}[\mathbf{C}]$ can be defined by taking morphisms $(A_1,A_1') \rightarrow (A_2,A_2')$ to be natural transformations $S:\mathbf{C}(A_1 \otimes -,A_1' \otimes =) \rightarrow \mathbf{C}(A_2 \otimes -,A_2' \otimes =)$ such that for every $T: \mathbf{C}(B_1 \otimes -  , B_1' \otimes =) \Rightarrow   \mathbf{C}(B_2 \otimes - , B_2' = ) $ then \[ \begin{tikzcd}
{\mathbf{C}(A_1 \otimes B_1 \otimes X,A_1' \otimes B_1' \otimes X')} \arrow[rr, "{S_{B_1 \otimes X,B_1' \otimes X'}}"] \arrow[d, "{\beta T \beta _{A_1,X,A_1',X'}}"'] &  & {\mathbf{C}(A_2 \otimes B_1 \otimes X,A_2' \otimes B_1' \otimes X')} \arrow[d, "{\beta T \beta_{A_2,X,A_2',X'}}"] \\
{\mathbf{C}(A_1 \otimes B_2 \otimes X,A_1' \otimes B_2' \otimes X')} \arrow[rr, "{S_{B_2 \otimes X,B_2' \otimes X'}}"']                                  &  & {\mathbf{C}(A_2 \otimes B_2 \otimes X,A_2' \otimes B_2' \otimes X')}                                 
\end{tikzcd} \]
\end{theorem}
\begin{proof}
From now on we omit indices on natural transformations. The assignment $\boxtimes$ given by
\begin{itemize}
    \item $[A,A']\boxtimes [B,B'] = [A \otimes B,A' \otimes B']$
    \item $(S \boxtimes T) = S \circ \beta \circ T \circ \beta$
\end{itemize}
defines a bifunctor $\boxtimes : \mathbf{slot}[\mathbf{C}] \times \mathbf{slot}[\mathbf{C}] \rightarrow \mathbf{slot}[\mathbf{C}]$. The interchange law is satisfied by the following 
\begin{align}
(S \boxtimes T)  (S' \boxtimes T') & = S\beta T \beta S' \beta T' \beta \\
& =    S S' \beta T \beta \beta T' \beta \\
& = (SS') \boxtimes (TT')
\end{align}
On the identity note that $i \boxtimes i = i \beta i \beta = \beta \beta = i$. The unit object is taken to be $(I,I)$,in the non-strict case one could define associators and unitors by inheriting them from $\mathbf{C}$. We assign a bifunctor $[-,-]$ by $[A,A'] := (A,A')$, with $[f,g]_{EE'}(\phi) := (g \otimes E') \circ \phi \circ (f \otimes E)$. The natural isomorphism is given by $\kappa(f)_{EE'}(\phi) = f \otimes \phi$ and $\kappa^{-1}(S) = S_{II}(id)$, where again we assume our underlying category is strict. The required morphism $p$ is given by the identity, which as a result immediately satisfies all of the relevant coherence conditions. 
\end{proof}

%%%%%%%%%%%%%%%%%%%%%%%%%%%

\section{Polycategory of polyslots}
To prove the following results algebraically is possible but extremely unreadable due to the need to keep track of symmetries, for readability we prefer to present our proofs in graphical form.
\begin{theorem}
The polyslots on $\mathbf{C}$ define a polycategory $\mathbf{pslot}[\mathbf{C}]$ with:
\begin{itemize}
    \item Objects given by pairs $[A,B]$ with $A,B$ objects of $\mathbf{C}$
    \item Poly-morphisms of type $S:\Gamma \rightarrow \Theta$ given by polyslots of type $S:[A_1,A_1'] \dots [A_n,A_n'] \rightarrow [B_1 \otimes \dots \otimes B_m ,B_1' \otimes  \dots \otimes B_m']$
    \item Composition given by  \[ \tikzfig{figs/locrepoly2} \]
    \item Symmetric action by permutations given by taking \[ \tikzfigscale{0.6}{figs/polyslotproof1} \] to be \[ \tikzfigscale{0.6}{figs/polyslotproof2} \]
\end{itemize}
\end{theorem}
\begin{proof}
We  confirm interchange laws for composition, that is, that: \[ \tikzfigscale{0.6}{figs/polyslotproof3} \quad = \quad \tikzfigscale{0.6}{figs/polyslotproof4} \] Indeed, consider \[ \tikzfigscale{0.6}{figs/polyslotproof5}, \] applying the symmetric action gives: \[ \tikzfigscale{0.6}{figs/polyslotproof6}, \] using the compsoition rule gives \[ \tikzfigscale{0.6}{figs/polyslotproof7}, \] or equivalently using the definition of slot induced by a polyslot \[ \tikzfigscale{0.6}{figs/polyslotproof8b}, \] We then use the composition rule again to give \[ \tikzfigscale{0.6}{figs/polyslotproof9}, \] Again converting into slot form gives: \[ \tikzfigscale{0.6}{figs/polyslotproof10}, \] using a series of swaps to set up the defining condition for slots gives \[ \tikzfigscale{0.6}{figs/polyslotproof11}, \] after-which the slot equation can finally be used to return \[ \tikzfigscale{0.6}{figs/polyslotproof12}, \]
Unpacking the definition of $\hat{T}$ gives
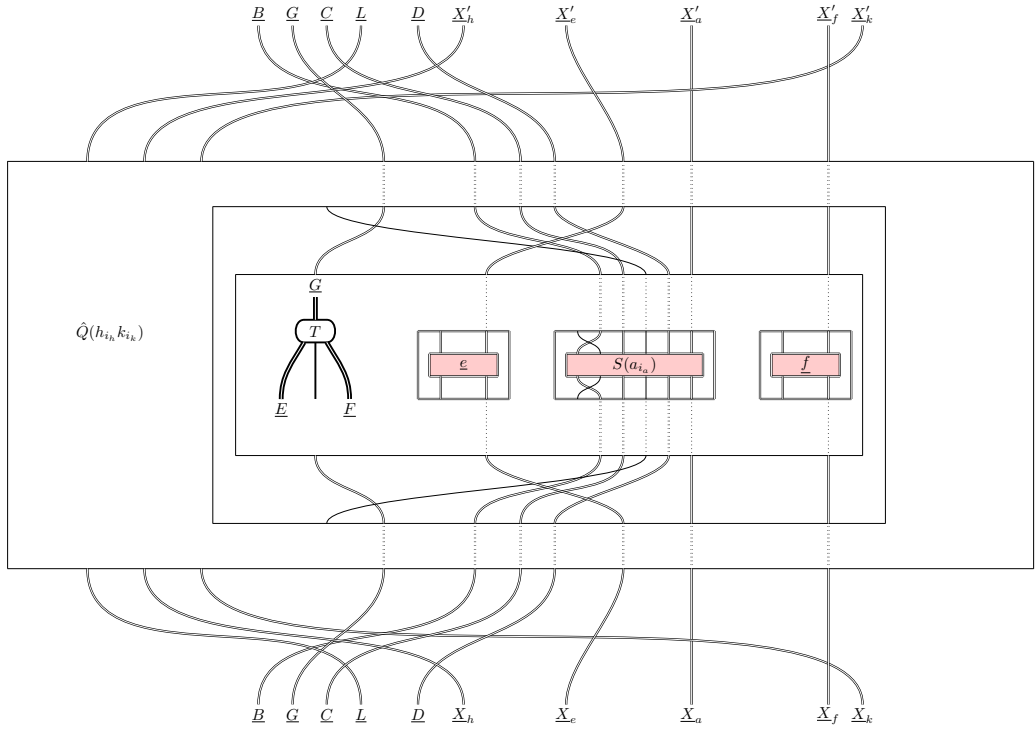
\[ \tikzfigscale{0.6}{figs/polyslotproof16}, \]
re-packaging the composition between $T$ and $S$ gives \[ \tikzfigscale{0.6}{figs/polyslotproof15}, \]
Unpacking the definition of $\hat{Q}$ gives
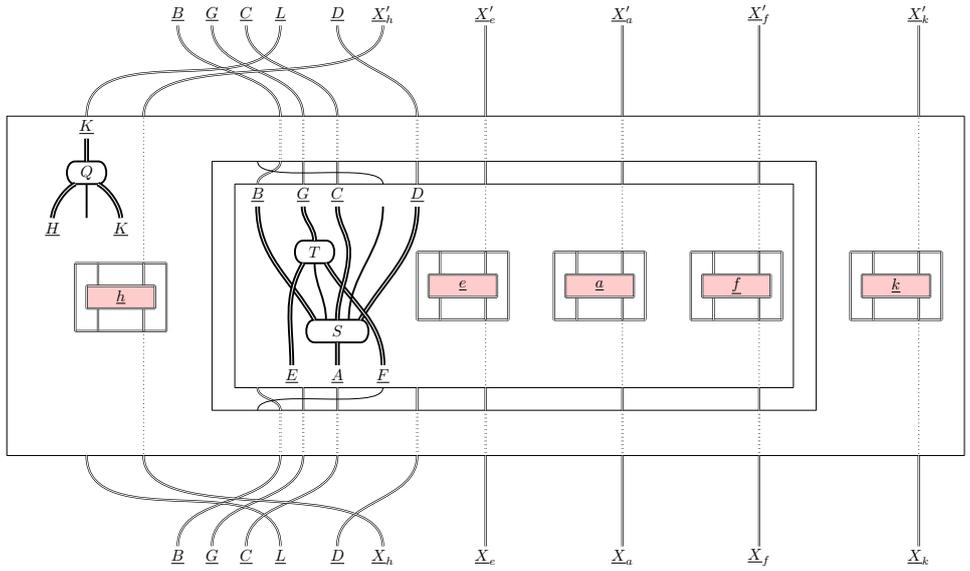
\[ \tikzfigscale{0.6}{figs/polyslotproof14}, \]
and finally repackaging the composition rule gives \[ \tikzfigscale{0.6}{figs/polyslotproof13}, \] and so indeed the interchange law is satisfied. The other interchange law which needs to be checked is more straightforward: \[ \tikzfigscale{0.6}{figs/polyslotproof3flip} \quad = \quad \tikzfigscale{0.6}{figs/polyslotproof4flip} \] We first consider the latter term,
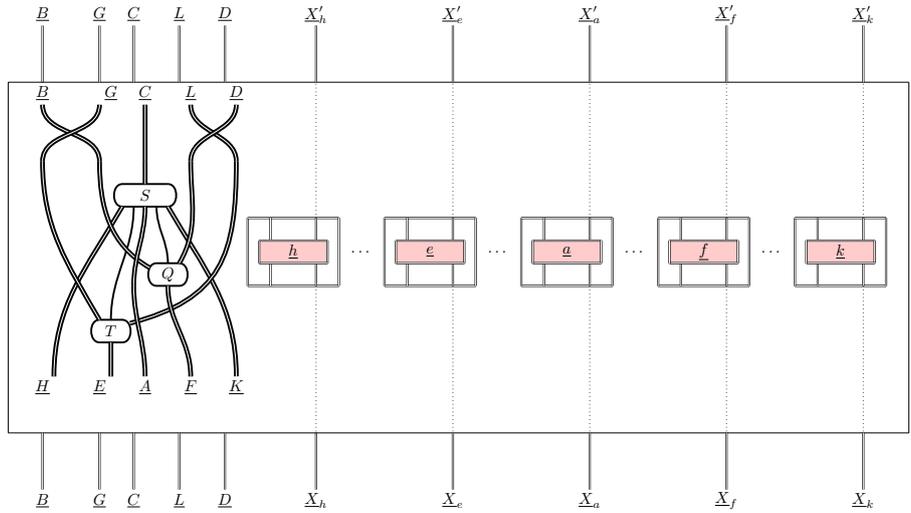
\[ \tikzfigscale{0.6}{figs/polyslotproof17}, \] and then use the definition of the symmetric action
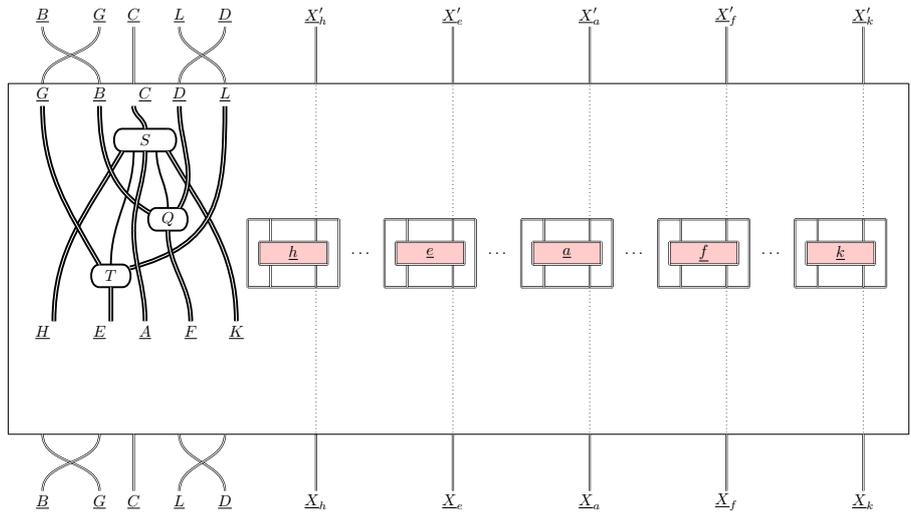
\[ \tikzfigscale{0.6}{figs/polyslotproof18}, \] then we use the definition of composition along $T$
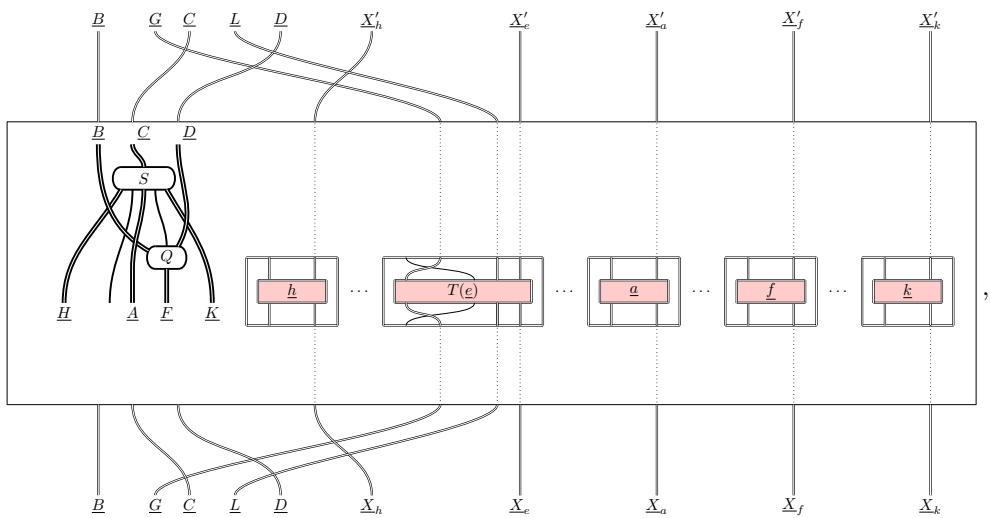
\[ \tikzfigscale{0.6}{figs/polyslotproof19b}, \] and then use the definition of composition along $Q$
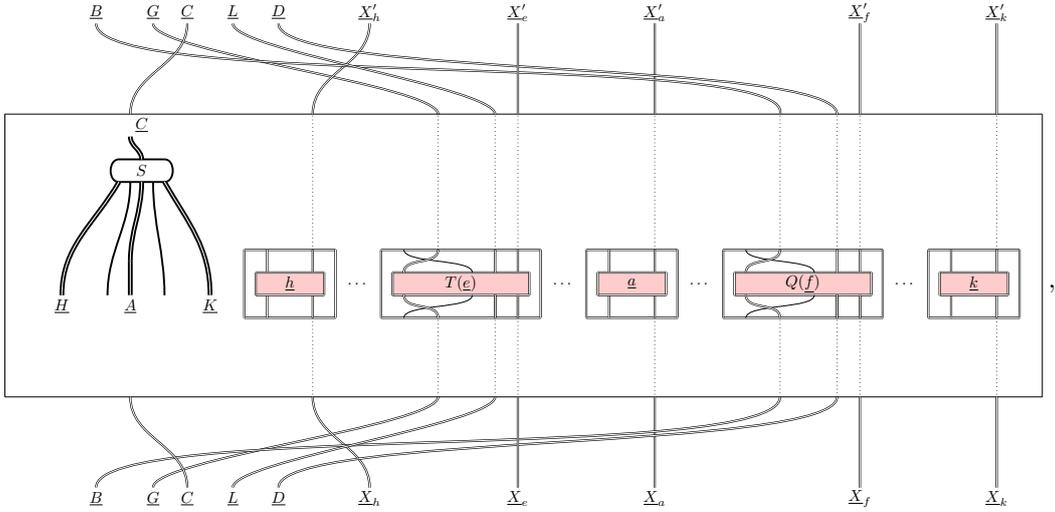
\[ \tikzfigscale{0.6}{figs/polyslotproof20}, \] then using the definition of composoition along $T$,
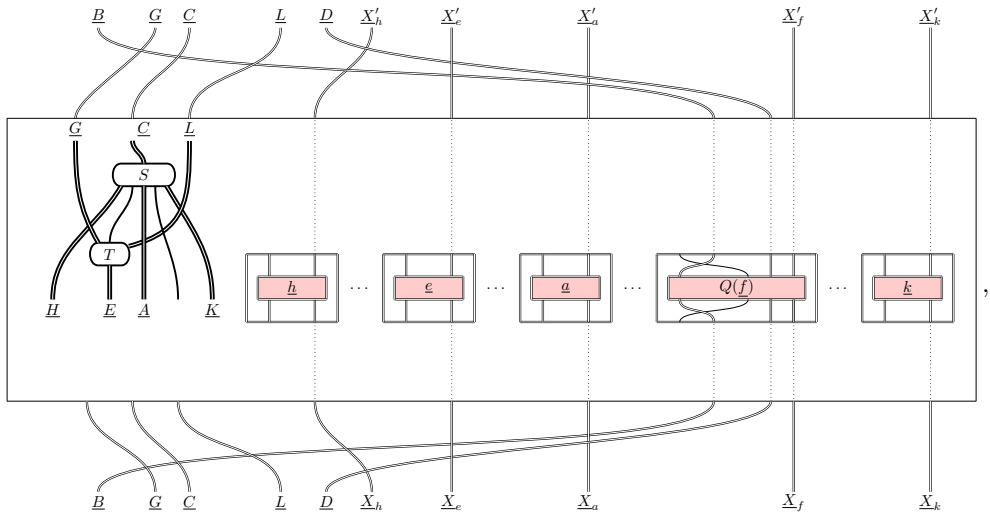
\[ \tikzfigscale{0.6}{figs/polyslotproof21}, \] and finally the definition of composition along $Q$ gives the result
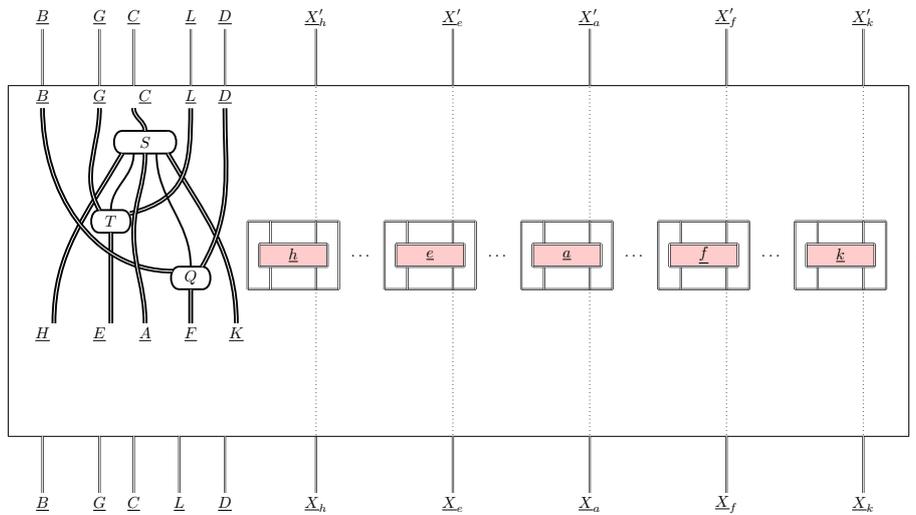
\[ \tikzfigscale{0.6}{figs/polyslotproof22}. \] the unit polymorphism of type $[A,A'] \rightarrow [A,A']$ is given by the slot with each $X,X'$ component given by the identity function of type $id: \mathbf{C}(AX,A'X')$. The associativity of sequential compositions is directly inherited from associativity of sequential composition for functions composition.
\end{proof}

%\begin{corollary}
%The single-party representable supermaps on $\mathbf{C}$ define a polycategory $\mathbf{srep}[\mathbf{C}]$.
%\end{corollary}
%\begin{proof}
%single-party representable supermaps are polyslots, so interchange laws need not be checked, all that needs to be checked is that local-representability is preserved under polyslot composition, which is follows by applying locally the observation that the composition of two combs is itself a comb.
%\end{proof}

%%%%%%%%%%%%%%%%%%%%%%%

\end{document}